\documentclass[reqno,10pt]{amsart}

\usepackage{amsaddr}
\usepackage{a4wide}
\usepackage[utf8]{inputenc}
\usepackage{ae,aecompl}
\usepackage[english]{babel}
\usepackage{enumerate}
\usepackage{latexsym}
\usepackage{amssymb}
\usepackage{amsfonts}
\usepackage{amsmath}
\usepackage{amsthm}
\usepackage{amsopn}
\usepackage{mathtools}
\usepackage{graphicx}
\usepackage{stmaryrd}  
\usepackage{braket}    
\usepackage{dsfont}
\usepackage{tikz}
\usepackage{enumitem}
\setlist{leftmargin=*}

\newtheorem{satz}{Theorem}[section]
\newtheorem{defi}[satz]{Defintion}
\newtheorem{bem}[satz]{Remark}

\newtheorem{prop}[satz]{Proposition} 
\newtheorem{cor}[satz]{Corollary}
\newtheorem{ass}[satz]{Assumption}
\numberwithin{equation}{section}

\newcommand{\erf}{\operatorname{erf}}

\newcommand{\sgn}{\operatorname{sgn}}
\newcommand{\AC}{\operatorname{AC}}
\newcommand{\Arg}{\operatorname{Arg}}
\renewcommand{\Re}{\operatorname{Re}}
\renewcommand{\Im}{\operatorname{Im}}

\newcommand{\interior}{\operatorname{int}}

\newcommand{\vvect}[2]{\left(\begin{array}{c} #1 \\ #2 \end{array}\right)}

\newcommand{\mmatrix}[4]{\left(\begin{array}{cc} #1 & #2 \\ #3 & #4 \end{array}\right)}

\title[Time evolution of superoscillations for the Schrödinger equation on $\mathbb{R}\setminus\{0\}$]{Time evolution of superoscillations \\ for the Schrödinger equation on $\mathbb{R}\setminus\{0\}$}
\author{Peter Schlosser}
\address{Graz University of Technology}

\begin{document}
 
\parindent 0pt

\begin{abstract}
In the context of quantum mechanics superoscillations, or the more general supershifts, appear as initial conditions of the time dependent Schrödinger equation. Already in \cite{ABCS21_2} a unified approach was developed, which yields time persistence of the supershift property under certain holomorphicity and growth assumptions on the corresponding Green's function. While that theory considers the Schrödinger equation on the whole real line $\mathbb{R}$, this paper takes the natural next step and considers $\mathbb{R}\setminus\{0\}$ instead, and allow boundary conditions at $x=0^\pm$ in addition. In particular the singular $\frac{1}{x^2}$-potential as well as the very important $\delta$ and $\delta'$ distributional potentials are covered.
\end{abstract}

\maketitle

\section{Introduction}

Superoscillations are functions with the paradoxical behaviour to (locally) oscillate faster than their largest Fourier component. The standard example which illustrates this behaviour is the sequence of functions
\begin{equation}\label{Eq_Example}
F_n(x)=\Big(\cos\Big(\frac{x}{n}\Big)+ik\sin\Big(\frac{x}{n}\Big)\Big)^n=\sum\limits_{j=0}^nC_je^{ik_jx},\qquad x\in\mathbb{R},
\end{equation}
with coefficients
\begin{equation}\label{Eq_Example_coefficients}
C_j=\binom{n}{m}\Big(\frac{1+k}{2}\Big)^{n-j}\Big(\frac{1-k}{2}\Big)^j\qquad\text{and}\qquad k_j=1-\frac{2j}{n},
\end{equation}
where $n\in\mathbb{N}_0$ and $k\in\mathbb{R}\setminus[-1,1]$. In particular, that $F$ is a certain linear combination of plane waves with frequencies $k_j\in[-1,1]$. The superocillatory behaviour now comes from the fact, that
\begin{equation}\label{Eq_Example_convergence_formal}
\lim\limits_{n\rightarrow\infty}F_n(x)=e^{ikx},\qquad x\in\mathbb{R},
\end{equation}
converges to a plane wave with frequency $|k|>1$. Note, that the convergence \eqref{Eq_Example_convergence_formal} is understood in the space $\mathcal{A}_1(\mathbb{C})$, see \cite[Lemma 2.4]{CSSY19_1} and Defintion \ref{defi_Ap}, which in particular implies uniform convergence on every compact subseteq of $\mathbb{R}$. What happens is an almost destructive interference of the plane waves $C_je^{ik_jx}$ with small frequencies $k_j\in[-1,1]$ but large amplitudes $C_j\sim|k|^n$, which leaves a remainder of the form $e^{ikx}$ with the small amplitude $1$ but high frequency $|k|>1$.

\medskip

In connection with quantum mechanics, these functions mainly come as the initial condition of the time dependent Schrödinger equation
\begin{subequations}
\begin{align}
i\frac{\partial}{\partial t}\Psi(t,x)&=\Big(-\frac{\partial^2}{\partial x^2}+V(t,x)\Big)\Psi(t,x), && t\in(0,T),\,x\in\mathbb{R}, \label{Eq_Schroedinger_R} \\
\Psi(0,x)&=F(x), && x\in\mathbb{R}. \label{Eq_Initial_R}
\end{align}
\end{subequations}
for some potential $V:(0,T)\times\mathbb{R}\rightarrow\mathbb{C}$. The question now is the superoscillatory behaviour of the solution $\Psi(t,x)$ at later times $t>0$.

\medskip

The first one addressing this problem was M. Berry who in \cite{B94_1} shows, that for free particles the superoscillatory behaviour occurs within a region $|x|<\mathcal{O}(n)$ and within a time $t<\mathcal{O}(n)$. In particular in the limit $n\rightarrow\infty$ this property is preserved everywhere and for all times. Thereafter, also for nonvanishing potentials the time persistence of superoscillations was proven, as for the harmonic oscillator in \cite{ACST18,ACSS20,ACSS18_2,BCSS14,BS15,CSSY19_1}, the electric field in \cite{ACSST17_2,ACST18,ACSS18_2,BCSS14}, the magnetic field in \cite{ACST18,CGS17}, the centrifugal potential in \cite{ACST18,ACSS20,CGS19,CSSY19_1}, the step potential in \cite{ACSST20_1} and distributional potentials as $\delta$ and $\delta'$ in \cite{ABCS20,ABCS21_1,BCS19}. It was also realized in \cite{CSSY19_1} that precise definition of superoscillations is for general potentials too narrow to persist in time. Hence superoscillations were generalized to supershifts as a consequence.

\medskip

However, up to this point only specific potentials were investigated, in particular cases where the corresponding Green's function was know explicitly. The first unified approach on the time persistence of supershifts, which only requires qualitative properties of the Green's function and no longer its explicit form, was given in \cite{ABCS21_2} for regular potentials $V(t,x)$, i.e. the Schrödinger equation \eqref{Eq_Schroedinger_R} is defined for all $x\in\mathbb{R}$. The topic of this paper is to continue this argument by considering the Schrödinger equation for $x\in\mathbb{R}\setminus\{0\}$ only, but additionally allow boundary conditions at $x=0^\pm$. I.e. we consider the time dependent Schrödinger equation
\begin{subequations}\label{Eq_Cauchy_R0}
\begin{align}
i\frac{\partial}{\partial t}\Psi(t,x)&=\Big(-\frac{\partial^2}{\partial x^2}+V(t,x)\Big)\Psi(t,x), && t\in(0,T),\,x\in\mathbb{R}\setminus\{0\}, \label{Eq_Schroedinger_R0} \\
M\vvect{\Psi(t,0^+)}{\Psi(t,0^+)}&=N\vvect{\frac{\partial}{\partial x}\Psi(t,0^+)}{-\frac{\partial}{\partial x}\Psi(t,0^-)}, && t\in(0,T), \label{Eq_Transmission_R0} \\
\Psi(0,x)&=F(x), && x\in\mathbb{R}\setminus\{0\}. \label{Eq_Initial_R0}
\end{align}
\end{subequations}
Where $V:(0,T)\times\mathbb{R}\setminus\{0\}\rightarrow\mathbb{C}$ is the potential and $M,N\in\mathbb{C}^{2\times 2}$ describe the boundary conditions at $x=0^\pm$. In particular, potentials with singularities, as for example $V(t,x)\sim\frac{1}{x^2}$, as well as distributional potentials as $\delta(x)$ or $\delta'(x)$ are covered by this approach. The key ingredient in the proof of the time persistence will be the representation
\begin{equation}\label{Eq_Greensfunction_integral_formal}
\Psi(t,x)=\int_\mathbb{R}G(t,x,y)F(y)dy
\end{equation}
via the corresponding Green's function $G$. The way how the integral \eqref{Eq_Greensfunction_integral_formal} will be interpreted is the main topic of the upcoming Section \ref{sec_Fresnel_integrals}.

\section{Fresnel integrals}\label{sec_Fresnel_integrals}

In this section we develop the so called Fresnel integral technique, which will be the way to interpret the integral \eqref{Eq_Greensfunction_integral_formal}. Roughly speaking, it is a method to make sense of integrals of the form
\begin{equation}\label{Eq_Nonintegrable}
\int_\mathbb{R}e^{iy^2}f(y)dy,
\end{equation}
also in situations where the function $f$ itself is not integrable. The basic idea is to use the Cauchy theorem to rotate the domain of integration into the complex plain and consequently make the oscillating prefactor $e^{iy^2}$ a Gaussian $e^{i(ye^{i\alpha})^2}$, whose decay at infinity ensures integrability.

\medskip

Note, that the subsequent Fresnel integral technique is in two ways an improvement of the version in \cite{ABCS21_2}. The first improvement lies in the fact, that we allow an exponential growth of order $p\in(0,2)$ in \eqref{Eq_Fresnel_boundedness} and \eqref{Eq_Fresnel_boundedness_R}, while in \cite{ABCS21_2} only $p=1$ was considered. The second improvement lies roughly speaking in the fact, that in \cite{ABCS21_2} the function $f$ had to be holomorphic in a neighborhood of the closed cone $S_\alpha^+\cup\{0\}$, in particular holomorphic in a neighborhood of $[0,\infty)$. In contrast, here it is enough for $f$ to be holomorphic in the interior of $S_\alpha^+$ with a continuous extension to $S_\alpha^+$.

\medskip

\begin{prop}[Fresnel integral]\label{prop_Fresnel_integral}
Let $a>0$, $x\in\mathbb{R}$. Consider for $\alpha\in(0,\frac{\pi}{2})$ the sector

\begin{minipage}{0.3\textwidth}
\hspace{1cm}
\begin{tikzpicture}
\fill[black!20] (2,1.48)--(0,0)--(2,0);
\draw[->] (-0.2,0)--(2.5,0) node[anchor=south] {\tiny{$\Re(z)$}};
\draw[->] (0,-0.2)--(0,1.5) node[anchor=west] {\tiny{$\Im(z)$}};
\draw[thick] (2,1.48)--(0,0)--(2,0);
\draw (1,0) arc (0:36.86:1);
\draw (0.4,0.2) node[anchor=west] {$\alpha$};
\draw (1.2,0.5) node[anchor=west] {$S_\alpha^+$};
\draw[fill=white] (0,0) circle (0.05);
\end{tikzpicture}
\end{minipage}
\begin{minipage}{0.69\textwidth}
\begin{equation}\label{Eq_Salpha+}
S_\alpha^+\coloneqq\Set{z\in\mathbb{C}\setminus\{0\} | \Arg(z)\in[0,\alpha]},
\end{equation}
\end{minipage}

and a continuous function $f:S_\alpha^+\rightarrow\mathbb{C}$ which is holomorphic on $\interior(S_\alpha^+)$ and satisfies the estimate
\begin{equation}\label{Eq_Fresnel_boundedness}
|f(z)|\leq Ae^{B|z|^p},\qquad z\in S_\alpha^+,
\end{equation}
for some and $A,B\geq 0$ and $p\in(0,2)$. Then for every $y_0\in\mathbb{R}$ we get
\begin{equation}\label{Eq_Fresnel_integral}
\lim\limits_{\varepsilon\rightarrow 0^+}\int_0^\infty e^{-\varepsilon(y-y_0)^2}e^{ia(y-x)^2}f(y)dy=e^{i\alpha}\int_0^\infty e^{ia(ye^{i\alpha}-x)^2}f(ye^{i\alpha})dy,
\end{equation}
where both integrands are absolute integrable. Moreover, for $0<\varepsilon<\frac{2a}{\tan(\alpha)}$ we also get
\begin{equation}\label{Eq_Fresnel_integral_without_limit}
\int_0^\infty e^{-\varepsilon(y-y_0)^2}e^{ia(y-x)^2}f(y)dy=e^{i\alpha}\int_0^\infty e^{-\varepsilon(ye^{i\alpha}-y_0)^2}e^{ia(ye^{i\alpha}-x)^2}f(ye^{i\alpha})dy.
\end{equation}
\end{prop}

\begin{proof}
Since the calculation is the same, we will for simplicity only consider $x=0$, $a=1$ and $y_0=0$. For any $\eta\in\interior(S_\alpha^+)$ with $|\eta|\leq 1$, we define the shifted function
\begin{equation}\label{Eq_Fresnel_integral_4}
f_\eta(z)\coloneqq f(z+\eta),\qquad z\in S_\alpha^+-\eta.
\end{equation}
Then $f_\eta$ is holomorphic on $\interior(S_\alpha^+)-\eta$ and, by \eqref{Eq_Fresnel_boundedness}, admits the exponential bound
\begin{equation}\label{Eq_Fresnel_integral_1}
|f_\eta(z)|\leq Ae^{B|z+\eta|^p}\leq Ae^{B2^p(|z|^p+|\eta|^p)}\leq Ae^{B2^p(|z|^p+1)}=\widetilde{A}e^{\widetilde{B}|z|^p},\qquad z\in S_\alpha^+-\eta,
\end{equation}
using the new constants $\widetilde{A}\coloneqq Ae^{B2^p}$ and $\widetilde{B}=B2^p$. Fixing $R>0$, we then consider the path

\begin{minipage}{0.45\textwidth}
\begin{align*}
\gamma_1&\coloneqq\Set{y | 0\leq y\leq R}, \\
\gamma_2&\coloneqq\Set{ye^{i\alpha} | 0\leq y\leq\frac{R}{\cos(\alpha)}}, \\
\gamma_3&\coloneqq\Set{R+iy | R\tan(\alpha)\geq y\geq 0}.
\end{align*}
\end{minipage}
\hspace{1cm}
\begin{minipage}{0.45\textwidth}
\begin{tikzpicture}[scale=0.7]
\fill[black!20] (3.5,-0.5)--(-1.6,-0.5)--(3.5,3.25);
\fill[black!40] (3.8,2.82)--(0,0)--(3.8,0);
\draw[thick] (3.5,-0.5)--(-1.6,-0.5)--(3.5,3.25);
\draw[thick] (0,0)--(3.8,2.82);
\draw[thick,->] (-1.5,0)--(5,0) node[anchor=south] {\tiny{$\Re(z)$}};
\draw[thick,->] (0,-0.8)--(0,3) node[anchor=west] {\tiny{$\Im(z)$}};
\draw[ultra thick,->] (0,0)--(2,0);
\draw[ultra thick] (0,0)--(3,0) (1.6,0.05) node[anchor=north] {$\gamma_1$};
\draw[ultra thick,->] (3,2)--(3,0.8);
\draw[ultra thick] (3,0)--(3,2.23) (2.95,1.2) node[anchor=west] {$\gamma_3$};
\draw[ultra thick,->] (0,0)--(2,1.48);
\draw[ultra thick] (0,0)--(3,2.23) (2,1.3) node[anchor=north] {$\gamma_2$};
\draw[thick] (1.3,0) arc (0:36.86:1.3) (0.5,0.25) node[anchor=west] {\small{$\alpha$}};
\draw[thick] (-0.3,-0.5) arc (0:36.86:1.3) (-1.2,-0.3) node[anchor=west] {\small{$\alpha$}};
\draw (3,0) node[anchor=north] {\tiny{$R$}};
\draw (3.4,1.7)--(4,1.5) node[anchor=west] {$S_\alpha^+$};
\draw (1.7,1.7)--(1,2) node[anchor=south] {$S_\alpha^+-\eta$};
\draw[fill=white] (-1.6,-0.5) circle (0.08) node[anchor=east] {$-\eta$};
\end{tikzpicture}
\end{minipage}

\medskip

Since the paths $\gamma_1$, $\gamma_2$, $\gamma_3$ lie inside $\interior(S_\alpha^+-\eta)$, where $f_\eta$ is holomorphic, Cauchy's theorem yields for every $\varepsilon>0$
\begin{equation}\label{Eq_Fresnel_integral_5}
\int_{\gamma_1}e^{(i-\varepsilon)z^2}f_\eta(z)dz=\int_{\gamma_2}e^{(i-\varepsilon)z^2}f_\eta(z)dz+\int_{\gamma_3}e^{(i-\varepsilon)z^2}f_\eta(z)dz.
\end{equation}
Using the exponential bound \eqref{Eq_Fresnel_integral_1}, we can estimate the integral along $\gamma_3$ as
\begin{align*}
\Big|\int_{\gamma_3}e^{(i-\varepsilon)z^2}f_\eta(z)dz\Big|&\leq\widetilde{A}e^{-\varepsilon R^2}\int_0^{R\tan(\alpha)}e^{\varepsilon y^2-2Ry+\widetilde{B}|R+iy|^p}dy \\
&\leq\widetilde{A}e^{-\varepsilon R^2+\frac{\widetilde{B}R^p}{\cos^p(\alpha)}}\int_0^{R\tan(\alpha)}e^{-y(2R-\varepsilon y)}dy \\
&\leq\widetilde{A}R\tan(\alpha)e^{-\varepsilon R^2+\frac{\widetilde{B}R^p}{\cos^p(\alpha)}},
\end{align*}
where in the last line we restricted $\varepsilon<\frac{2}{\tan(\alpha)}$ to conclude $2R-\varepsilon y>0$. This estimate proves the convergence
\begin{equation*}
\lim\limits_{R\rightarrow\infty}\int_{\gamma_3}e^{(i-\varepsilon)z^2}f_\eta(z)dz=0.
\end{equation*}
Consequently, in the limit $R\rightarrow\infty$, the integrals \eqref{Eq_Fresnel_integral_5} become
\begin{equation}\label{Eq_Fresnel_integral_11}
\int_0^\infty e^{(i-\varepsilon)y^2}f_\eta(y)dy=e^{i\alpha}\int_0^\infty e^{(i-\varepsilon)(ye^{i\alpha})^2}f_\eta(ye^{i\alpha})dy.
\end{equation}
Here both integrals are absolute convergent, the left hand side because of the factor $e^{-\varepsilon y^2}$ and the right hand side due to the estimate
\begin{align}
\big|e^{(i-\varepsilon)(ye^{i\alpha})^2}f_\eta(ye^{i\alpha})\big|&\leq\widetilde{A}e^{-(\sin(2\alpha)+\varepsilon\cos(2\alpha))y^2+\widetilde{B}y^p}, \notag \\
&\leq\widetilde{A}e^{-(\frac{2}{\tan(\alpha)}-\varepsilon)\sin^2(\alpha)y^2+\widetilde{B}y^p}, \label{Eq_Fresnel_integral_2}
\end{align}
which is integrable for every $\varepsilon<\frac{2}{\tan(\alpha)}$. Moreover, since the upper bound \eqref{Eq_Fresnel_integral_2} is $\eta$-independent, we can apply the dominated convergence theorem to both sides of \eqref{Eq_Fresnel_integral_11} and obtain
\begin{equation}\label{Eq_Fresnel_integral_3}
\int_0^\infty e^{(i-\varepsilon)y^2}f(y)dy=e^{i\alpha}\int_0^\infty e^{(i-\varepsilon)(ye^{i\alpha})^2}f(ye^{i\alpha})dy,
\end{equation}
which is exactly the identity \eqref{Eq_Fresnel_integral_without_limit}. Finally, we want to apply the limit $\varepsilon\rightarrow 0^+$ to this equation. By the estimate \eqref{Eq_Fresnel_integral_2} for $f$ instead of $f_\eta$, i.e. formally putting $\eta=0$, the integrand on the right hand side of \eqref{Eq_Fresnel_integral_3} is bounded by some majorant which decreases as $\varepsilon\rightarrow 0^+$. By the dominated convergence theorem we then obtain the stated limit \eqref{Eq_Fresnel_integral}.
\end{proof}

The Fresnel integral technique of Proposition \ref{prop_Fresnel_integral} can also be applied on the negative semi axis, which leads to the following corollary.

\medskip

\begin{cor}\label{cor_Fresnel_integral}
Let $a>0$, $x\in\mathbb{R}$. Consider for some $\alpha\in(0,\frac{\pi}{2})$ the double sector

\begin{minipage}{0.3\textwidth}
\begin{center}
\begin{tikzpicture}
\fill[black!20] (0.05,0)--(0.04,0.03)--(1.5,1.11)--(1.5,0);
\fill[black!20] (-0.05,0)--(-0.04,-0.03)--(-1.5,-1.11)--(-1.5,0);
\draw[->] (-2,0)--(2,0) node[anchor=south] {\tiny{$\Re(z)$}};
\draw[->] (0,-1)--(0,1) node[anchor=west] {\tiny{$\Im(z)$}};
\draw[thick] (-1.5,-1.11)--(1.5,1.11);
\draw[thick] (-1.5,0)--(1.5,0);
\draw (0.8,0) arc (0:36.86:0.8) (0.3,0.18) node[anchor=west] {$\alpha$};
\draw (-0.8,0) arc (180:216.86:0.8) (-0.3,-0.18) node[anchor=east] {$\alpha$};
\draw (0.8,0.4) node[anchor=west] {\large{$S_\alpha$}};
\draw[fill=white] (0,0) circle (0.05);
\end{tikzpicture}
\end{center}
\end{minipage}
\begin{minipage}{0.65\textwidth}
\begin{equation}\label{Eq_Salpha}
S_\alpha\coloneqq\Set{z\in\mathbb{C}\setminus\{0\} | \Arg(z)\in[0,\alpha]\cup[\pi,\pi+\alpha]},
\end{equation}
\end{minipage}

and a continuous function $f:S_\alpha\rightarrow\mathbb{C}$ which is holomporphic on $\interior(S_\alpha)$ and satisfies the estimate
\begin{equation}\label{Eq_Fresnel_boundedness_R}
|f(z)|\leq Ae^{B|z|^p},\qquad z\in S_\alpha,
\end{equation}
for some $A,B\geq 0$ and $p\in(0,2)$. Then for every $y_0\in\mathbb{R}$ we get
\begin{equation}\label{Eq_Fresnel_integral_R}
\lim\limits_{\varepsilon\rightarrow 0^+}\int_\mathbb{R}e^{-\varepsilon(y-y_0)^2}e^{ia(y-x)^2}f(y)dy=e^{i\alpha}\int_\mathbb{R}e^{ia(ye^{i\alpha}-x)^2}f(ye^{i\alpha})dy,
\end{equation}
where both integrands are absolute integrable. Moreover, for $0<\varepsilon<\frac{2a}{\tan(\alpha)}$ we also get
\begin{equation}\label{Eq_Fresnel_integral_R_without_limit}
\int_\mathbb{R}e^{-\varepsilon(y-y_0)^2}e^{ia(y-x)^2}f(y)dy=e^{i\alpha}\int_\mathbb{R}e^{-\varepsilon(ye^{i\alpha}-y_0)^2}e^{ia(ye^{i\alpha}-x)^2}f(ye^{i\alpha})dy.
\end{equation}
\end{cor}

\section{Schrödinger equation on $\mathbb{R}\setminus\{0\}$}\label{sec_Schroedinger_equation_on_R0}

The central topic of this paper is the investigation of the Cauchy problem \eqref{Eq_Cauchy_R0}. In particular, we consider the Green's function approach \eqref{Eq_Greensfunction_integral_formal}. The main result of this section will then be Theorem \ref{satz_Greensfunction}, which puts the integral \eqref{Eq_Greensfunction_integral_formal} into a mathematical rigorous framework and also provides a continuous dependency between the initial condition $F$ and the solution $\Psi$. This continuous dependency result will then be the main ingredient in Section \ref{sec_Stability_of_superoscillations_and_supershifts} to conclude the timer persistence of the supershift property.

\medskip

We start by specifying in detail in which sense we want to understand the Cauchy problem \eqref{Eq_Cauchy_R0}. It will be convenient to view the solution (and its derivatives) in the context of absolute continuous functions. The linear space of absolute continuous functions on some open interval $I\subseteq\mathbb{R}$ will be denoted by $\AC(I)$. Recall, that $f\in\AC(I)$ if and only if there exists some $g\in L^1_\text{loc}(I)$, such that
\begin{equation}\label{Eq_Absolute_continuous_integral}
f(y)-f(x)=\int_x^y g(s)ds,\qquad x,y\in I.
\end{equation}
Also observe, that $f\in\AC(I)$ is differentiable almost everywhere and its derivative $f'$ coincides almost everywhere with $g$ in \eqref{Eq_Absolute_continuous_integral}. Moreover, we understand the space of absolute continuous functions on $\dot{\mathbb{R}}\coloneqq\mathbb{R}\setminus\{0\}$ as
\begin{equation*}
\AC(\dot{\mathbb{R}})\coloneqq\Set{f:\dot{\mathbb{R}}\rightarrow\mathbb{C} | f\vert_{(-\infty,0)}\in\AC((-\infty,0))\text{ and }f\vert_{(0,\infty)}\in\AC((0,\infty))}.
\end{equation*}
For $T\in(0,\infty]$ we shall now work with the space
\begin{equation}\label{Eq_AC12_R0}
\AC_{1,2}((0,T)\times\dot{\mathbb{R}})\coloneqq\Set{\Psi:(0,T)\times\dot{\mathbb{R}}\rightarrow\mathbb{C} | \begin{array}{l} \Psi(\,\cdot\,,x)\in\AC((0,T)),\;\forall x\in\dot{\mathbb{R}} \\ \Psi(t,\,\cdot\,),\Psi_x(t,\,\cdot\,)\in\AC(\dot{\mathbb{R}}),\;\forall t\in(0,T) \end{array}}.
\end{equation}
Let $V:(0,T)\times\dot{\mathbb{R}}\rightarrow\mathbb{C}$ be some potential, $M,N\in\mathbb{C}^{2\times 2}$ matrices describing the transmission condition and $F:\dot{\mathbb{R}}\rightarrow\mathbb{C}$ the initial condition. We call a function $\Psi\in\AC_{1,2}((0,T)\times\dot{\mathbb{R}})$ a solution of the time dependent Schrödinger equation, if it satisfies
\begin{subequations}\label{Eq_Psi_Cauchy}
\begin{align}
i\frac{\partial}{\partial t}\Psi(t,x)&=\Big(-\frac{\partial^2}{\partial x^2}+V(t,x)\Big)\Psi(t,x), && \text{f.a.e. }t\in(0,T),\,x\in\dot{\mathbb{R}}, \label{Eq_Psi_Schroedinger} \\
M\vvect{\Psi(t,0^+)}{\Psi(t,0^+)}&=N\vvect{\Psi_x(t,0^+)}{-\Psi_x(t,0^-)}, && t\in(0,T), \label{Eq_Psi_transmission} \\
\lim\limits_{t\rightarrow 0^+}\Psi(t,x)&=F(x), && x\in\dot{\mathbb{R}}. \label{Eq_Psi_initial}
\end{align}
\end{subequations}
The corresponding \textit{Green's function} is a function $G:(0,T)\times\dot{\mathbb{R}}\times\dot{\mathbb{R}}\rightarrow\mathbb{C}$, which depends on the potential $V$ and the boundary matrices $M,N$, but not on the initial condition $F$, such that the solution $\Psi$ admits the (formal) representation
\begin{equation}\label{Eq_Psi_integral_formal}
\Psi(t,x)=\int_\mathbb{R}G(t,x,y)F(y)dy,\qquad t\in(0,T),\,x\in\dot{\mathbb{R}}.
\end{equation}
In the following Assumptions \ref{ass_Greensfunction} we provide a set of properties for the Green's function $G$, in order to give meaning to the integral \eqref{Eq_Psi_integral_formal} and ensure that $\Psi(t,x)$ indeed is a solution of the Cauchy problem \eqref{Eq_Psi_Cauchy}.

\medskip

\begin{ass}\label{ass_Greensfunction}
Let $T\in(0,\infty]$ and $G:(0,T)\times\dot{\mathbb{R}}\times\dot{\mathbb{R}}\rightarrow\mathbb{C}$. For some $\alpha\in(0,\frac{\pi}{2})$ let $S_\alpha$ be the double sector \eqref{Eq_Salpha}, and suppose that $G$ admits a continuation $G:(0,T)\times\dot{\mathbb{R}}\times S_\alpha\rightarrow\mathbb{C}$, such that for every fixed $t\in(0,T)$, $x\in\dot{\mathbb{R}}$ the mapping $G(t,x,\,\cdot\,)$ is continuous on $S_\alpha$ and holomorphic on $\interior(S_\alpha)$. Moreover, it will be assumed that $G$ satisfies the following properties (i)--(iii).

\begin{enumerate}
\item[(i)] For every fixed $z\in S_\alpha$, the function $G(\,\cdot\,,\,\cdot\,,z)\in\AC_{1,2}((0,T)\times\dot{\mathbb{R}})$ is a solution of the time dependent Schrödinger equation
\begin{equation}\label{Eq_G_Schroedinger}
i\frac{\partial}{\partial t}G(t,x,z)=\Big(-\frac{\partial^2}{\partial x^2}+V(t,x)\Big)G(t,x,z),\qquad\text{f.a.e. }t\in(0,T),\,x\in\dot{\mathbb{R}},
\end{equation}
with $V:(0,T)\times\dot{\mathbb{R}}\rightarrow\mathbb{C}$ the considered potential. Moreover, for every $y\in\dot{\mathbb{R}}$ the Green's function satisfies the transmission condition
\begin{equation}\label{Eq_G_transmission}
M\vvect{G(t,0^+,y)}{G(t,0^-,y)}=N\vvect{G_x(t,0^+,y)}{-G_x(t,0^-,y)},\qquad t\in(0,T),
\end{equation}
with matrices $M,N\in\mathbb{C}^{2\times 2}$.

\item[(ii)] For every $x\in\dot{\mathbb{R}}$ there exists some $x_0>|x|$, such that
\begin{equation}\label{Eq_G_initial}
\lim\limits_{t\rightarrow 0^+}\int_{-x_0}^{x_0}G(t,x,y)\varphi(y)dy=\varphi(x),\qquad\varphi\in\mathcal{C}^\infty([-x_0,x_0]).
\end{equation}

\item[(iii)] There exists $a\in\AC((0,T))$ with $a(t)>0$ and $\lim_{t\rightarrow 0^+}a(t)=\infty$, such that the function $\widetilde{G}$ in the decomposition
\begin{equation}\label{Eq_G_decomposition}
G(t,x,z)=e^{ia(t)(z-x)^2}\widetilde{G}(t,x,z),\qquad t\in(0,T),\,x\in\dot{\mathbb{R}},\,z\in S_\alpha,
\end{equation}
is for every $t\in(0,T)$, $x\in\dot{\mathbb{R}}$ exponentially bounded as
\begin{subequations}\label{Eq_Gtilde_estimates}
\begin{align}
\big|\widetilde{G}(t,x,z)\big|&\leq A_0(t,x)e^{B_0(t,x)|z|^p},\qquad z\in S_\alpha, \label{Eq_Gtilde_estimate} \\
\Big|\frac{\partial}{\partial x}\widetilde{G}(t,x,z)\Big|&\leq A_1(t,x)e^{B_1(t,x)|z|^p},\qquad z\in S_\alpha,\label{Eq_Gtildex_estimate} \\
\Big|\frac{\partial^2}{\partial x^2}\widetilde{G}(t,x,z)\Big|,\,\Big|\frac{\partial}{\partial t}\widetilde{G}(t,x,z)\Big|&\leq A_2(t,x)e^{B_2(t,x)|z|^p},\qquad z\in S_\alpha.\label{Eq_Gtilde_derivative_estimate}
\end{align}
\end{subequations}
Here $p\in(0,2)$ and $A_0,A_1,A_2,B_0,B_1,B_2:(0,T)\times\dot{\mathbb{R}}\rightarrow[0,\infty)$ are continuous and for every $x\in\dot{\mathbb{R}}$
\begin{equation}\label{Eq_Coefficient_extension}
\frac{A_0(\,\cdot\,,x)}{\sqrt{a(t)}}\text{ and }B_0(\,\cdot\,,x)\text{ are bounded as }t\rightarrow 0^+,
\end{equation}
and for every $t\in(0,T)$:

\begin{enumerate}
\item[$\circ$] If $M=N=0$, no further assumptions.

\item[$\circ$] If $M\neq 0$, $N=0$, then $A_0(t,\,\cdot\,)$, $B_0(t,\,\cdot\,)$ are bounded as $x\rightarrow 0^\pm$.

\item[$\circ$] If $N\neq 0$, then $A_0(t,\,\cdot\,)$, $A_1(t,\,\cdot\,)$, $B_0(t,\,\cdot\,)$, $B_1(t,\,\cdot\,)$ are bounded as $x\rightarrow 0^\pm$.
\end{enumerate}
\end{enumerate}
\end{ass}

\medskip

Once we fixed the assumptions on the Green's function, we still need to specify the allowed initial conditions $F$ in \eqref{Eq_Psi_initial}. The following space $\mathcal{A}_q(\mathbb{C})$ is natural in the sense that it fits with the assumptions \eqref{Eq_Fresnel_boundedness_R} and also contains the superoscillating functions \eqref{Eq_Example} for $p=1$ as well as the supershift functions of Definition \ref{defi_Supershift}.

\medskip

\begin{defi}\label{defi_Ap}
Let $\mathcal{H}(\mathbb{C})$ denote the set of all entire functions. Then for every $q>0$ define the \textit{space of entire functions with exponential growth of order} $q$ as
\begin{equation}\label{Eq_Ap_space}
\mathcal{A}_q(\mathbb{C})\coloneqq\Set{F\in\mathcal{H}(\mathbb{C}) | \exists A,B\geq 0\text{ such that }|F(z)|\leq Ae^{B|z|^q}\text{ for all }z\in\mathbb{C}}.
\end{equation}
A sequence of functions $(F_n)_n\in\mathcal{A}_q(\mathbb{C})$ converges to $F_0\in\mathcal{A}_q(\mathbb{C})$ in $\mathcal{A}_q(\mathbb{C})$, if and only if there exists some $B\geq 0$, such that
\begin{equation}\label{Eq_Ap_convergence}
\lim\limits_{n\rightarrow\infty}\sup\limits_{z\in\mathbb{C}}|F_n(z)-F_0(z)|e^{-B|z|^q}=0.
\end{equation}
\end{defi}

\begin{satz}\label{satz_Greensfunction}
Let $G:(0,T)\times\dot{\mathbb{R}}\times\dot{\mathbb{R}}\rightarrow\mathbb{C}$ be as in Assumption \ref{ass_Greensfunction}. Then for every $F\in\mathcal{A}_q(\mathbb{C})$, $q\in(0,2)$, the wave function
\begin{equation}\label{Eq_Psi}
\Psi(t,x)\coloneqq\lim\limits_{\varepsilon\rightarrow 0^+}\int_\mathbb{R}e^{-\varepsilon y^2}G(t,x,y)F(y)dy,\qquad t\in(0,T),\,x\in\dot{\mathbb{R}},
\end{equation}
exists and $\Psi\in\AC_{1,2}((0,T)\times\dot{\mathbb{R}})$ is a solution of the Cauchy problem \eqref{Eq_Psi_Cauchy}. Moreover, if initial conditions $(F_n)_n\in\mathcal{A}_q(\mathbb{C})$ converge as $F_n\overset{n\rightarrow\infty}{\longrightarrow}F$ in $\mathcal{A}_q(\mathbb{C})$, also the corresponding solutions converge as
\begin{equation}\label{Eq_Psi_convergence}
\lim\limits_{n\rightarrow\infty}\Psi(t,x;F_n)=\Psi(t,x;F),
\end{equation}
for fixed $t\in(0,T)$ and uniformly on compact subsets of $\dot{\mathbb{R}}$.
\end{satz}

\medskip

Note, that for convenience we used the notation $\Psi(t,x;F)$ to emphasize the initial condition.

\medskip

Since in practical applications the initial condition \eqref{Eq_G_initial} is often hard to verify, the following Corollary \ref{cor_Greensfunction} gives an opportunity to replace it by the simple limit \eqref{Eq_G_initial_easy}. Roughly speaking, the limit \eqref{Eq_G_initial_easy} is one way how the Green's function approaches the $\delta$-function as $t\rightarrow 0^+$. However, in order to use this simplification it is necessary for the Green's function to be holomorphic (and satisfy \eqref{Eq_Gtilde_estimate}) not only on $S_\alpha$ but also on a neighborhood of $\dot{\mathbb{R}}$. More precisely, for $\alpha\in(0,\frac{\pi}{2})$ and $h>0$ we consider

\begin{minipage}{0.19\textwidth}
\begin{tikzpicture}[scale=0.7]
\fill[black!20] (2,-0.5)--(0.68,-0.5)--(0.04,-0.03)--(0.05,0)--(0.04,0.03)--(2,1.48);
\fill[black!20] (-2,0.5)--(-0.68,0.5)--(-0.04,0.03)--(-0.05,0)--(-0.04,-0.03)--(-2,-1.48);
\draw[->] (-2.5,0)--(2.5,0) node[anchor=south] {\tiny{$\Re$}};
\draw[->] (0,-1.5)--(0,1.5) node[anchor=west] {\tiny{$\Im$}};;
\draw[thick] (-2,-1.48)--(2,1.48);
\draw[thick] (-2,0.5)--(-0.68,0.5)--(0.68,-0.5)--(2,-0.5);
\draw (0.68,-0.5) arc (-36.86:36.86:0.84) (0.3,0.2) node[anchor=west] {\scriptsize{$\alpha$}};
\draw (-0.68,0.5) arc (143.14:216.86:0.84) (-0.3,0.2) node[anchor=east] {\scriptsize{$\alpha$}};
\draw (0.3,-0.2) node[anchor=west] {\scriptsize{$\alpha$}};
\draw (-0.3,-0.2) node[anchor=east] {\scriptsize{$\alpha$}};
\draw (-0.1,0.5)--(0.1,0.5) (0,0.5) node[anchor=east] {\scriptsize{$h$}};
\draw (-0.1,-0.5)--(0.1,-0.5) (0,-0.5) node[anchor=east] {\scriptsize{-$h$}};
\draw (0.8,0.35) node[anchor=west] {\small{$D_{\alpha,h}$}};
\draw[fill=white] (0,0) circle (0.08);
\end{tikzpicture}
\end{minipage}
\begin{minipage}{0.79\textwidth}
\begin{align}
D_{\alpha,h}\coloneqq&\Set{z\in\mathbb{C}\setminus\{0\} | \Arg(z)\in[-\alpha,\alpha]\text{ and }\Im(z)\geq -h} \label{Eq_Dalphah} \\
&\cup\Set{z\in\mathbb{C}\setminus\{0\} | \Arg(z)\in[\pi-\alpha,\pi+\alpha]\text{ and }\Im(z)\leq h}. \notag \\ \notag
\end{align}
\end{minipage}

\begin{cor}\label{cor_Greensfunction}
Let $G:(0,T)\times\dot{\mathbb{R}}\times\dot{\mathbb{R}}\rightarrow\mathbb{C}$ satisfy the Assumption \ref{ass_Greensfunction} with $S_\alpha$ replaced by $D_{\alpha,h}$ and \eqref{Eq_G_initial} replaced by
\begin{equation}\label{Eq_G_initial_easy}
\lim\limits_{t\rightarrow 0^+}\frac{G(t,x,x)}{\sqrt{a(t)}}=\frac{1}{\sqrt{i\pi}},\qquad x\in\dot{\mathbb{R}}.
\end{equation}
Then the same results as in Theorem \ref{satz_Greensfunction} hold true. I.e. for every $F\in\mathcal{A}_q(\mathbb{C})$, $q\in(0,2)$, the wave function \eqref{Eq_Psi} exists and $\Psi\in\AC_{1,2}((0,T)\times\dot{\mathbb{R}})$ is a solution of the Cauchy problem \eqref{Eq_Psi_Cauchy}. Moreover, if initial conditions $(F_n)_n\in\mathcal{A}_q(\mathbb{C})$ converge as $F_n\overset{n\rightarrow\infty}{\longrightarrow}F$ in $\mathcal{A}_q(\mathbb{C})$, also the corresponding solutions converge as in \eqref{Eq_Psi_convergence}.
\end{cor}

\begin{proof}[Proof of Theorem \ref{satz_Greensfunction}]
First we note, that due to $F\in\mathcal{A}_q(\mathbb{C})$, there exists $A,B\geq 0$ such that
\begin{equation}\label{Eq_F_estimate}
|F(z)|\leq Ae^{B|z|^q},\qquad z\in S_\alpha.
\end{equation}
\textit{Step 1.} In the first step we apply Corollary \ref{cor_Fresnel_integral}, to show that the expression \eqref{Eq_Psi} for the wave function is meaningful and give a representation using Fresnel integrals. For this, we fix $t\in(0,T)$, $x\in\dot{\mathbb{R}}$ and use the estimates \eqref{Eq_Gtilde_estimate} and \eqref{Eq_F_estimate} to get
\begin{align}
|\widetilde{G}(t,x,z)F(z)|&\leq AA_0(t,x)e^{B_0(t,x)|z|^p+B|z|^q} \notag \\
&\leq AA_0(t,x)e^{(B+B_0(t,x))(1+|z|)^{\max\{p,q\}}} \notag \\
&\leq AA_0(t,x)e^{(B+B_0(t,x))2^{\max\{p,q\}}(1+|z|^{\max\{p,q\}})} \notag \\
&=\widetilde{A}_0(t,x)e^{\widetilde{B}_0(t,x)|z|^{\widetilde{p}}},\qquad z\in S_\alpha, \label{Eq_Estimate_Integrand}
\end{align}
where we introduced the new coefficients
\begin{align}
\widetilde{A}_0(t,x)&\coloneqq AA_0(t,x)e^{(B+B_0(t,x))2^{\max\{p,q\}}}, \notag \\
\widetilde{B}_0(t,x)&\coloneqq(B+B_0(t,x))2^{\max\{p,q\}}, \label{Eq_Tilde_coefficients} \\
\widetilde{p}&\coloneqq\max\{p,q\}. \notag
\end{align}
Hence, due to the decomposition \eqref{Eq_G_decomposition}, the assumptions of Corollary \ref{cor_Fresnel_integral} are satisfied, which means that the wave function \eqref{Eq_Psi} exists and admits the absolute integrable representation
\begin{align}
\Psi(t,x)&=\lim\limits_{\varepsilon\rightarrow 0^+}\int_\mathbb{R}e^{-\varepsilon y^2}e^{ia(t)(y-x)^2}\widetilde{G}(t,x,y)F(y)dy \notag \\
&=e^{i\alpha}\int_\mathbb{R}e^{ia(t)(ye^{i\alpha}-x)^2}\widetilde{G}(t,x,ye^{i\alpha})F(ye^{i\alpha})dy \notag \\
&=e^{i\alpha}\int_\mathbb{R}G(t,x,ye^{i\alpha})F(ye^{i\alpha})dy. \label{Eq_Psi_Fresnel}
\end{align}
\textit{Step 2.} We show that the function $\Psi$ in \eqref{Eq_Psi}, is a solution of the Schrödinger equation \eqref{Eq_Psi_Schroedinger}. Roughly speaking, since $G$ is already a solution of \eqref{Eq_G_Schroedinger} by Assumption \ref{ass_Greensfunction} (i), it is sufficient to carry the derivatives inside the integral \eqref{Eq_Psi_Fresnel}.

\medskip

For the first spatial derivative we note, that $G(t,\,\cdot\,,z)\in\AC(\dot{\mathbb{R}})$ for every $t\in(0,T)$, $z\in S_\alpha$ by Assumption \ref{ass_Greensfunction} (i). Hence, for any $x_0,x_1>0$ we have
\begin{equation*}
G(t,x_1,z)=G(t,x_0,z)+\int_{x_0}^{x_1}\frac{\partial}{\partial x}G(t,x,z)dx,\qquad t\in(0,T),\,z\in S_\alpha,
\end{equation*}
which leads to the following integral representation of the wave function \eqref{Eq_Psi_Fresnel}
\begin{equation}\label{Eq_Psix}
\Psi(t,x_1)=\Psi(t,x_0)+e^{i\alpha}\int_\mathbb{R}\int_{x_0}^{x_1}\frac{\partial}{\partial x}G(t,x,ye^{i\alpha})dxF(ye^{i\alpha})dy.
\end{equation}
Using the decomposition \eqref{Eq_G_decomposition}, we can write the derivative as
\begin{equation*}
\frac{\partial}{\partial x}G(t,x,ye^{i\alpha})=\Big(2ia(t)(x-ye^{i\alpha})\widetilde{G}(t,x,ye^{i\alpha})+\frac{\partial}{\partial x}\widetilde{G}(t,x,ye^{i\alpha})\Big)e^{ia(t)(ye^{i\alpha}-x)^2}.
\end{equation*}
Using the estimate \eqref{Eq_Estimate_Integrand} and a similar one for $\frac{\partial}{\partial x}\widetilde{G}(t,x,z)F(z)$ using the coefficients $\widetilde{A}_1(t,x)\coloneqq AA_1(t,x)e^{(B+B_1(t,x))2^{\max\{p,q\}}}$ and $\widetilde{B}_1(t,x)\coloneqq(B+B_1(t,x))2^{\max\{p,q\}}$, we get
\begin{align}
&\Big|\frac{\partial}{\partial x}G(t,x,ye^{i\alpha})F(ye^{i\alpha})\Big| \notag \\
&\hspace{1.5cm}=\Big|2ia(t)(x-ye^{i\alpha})\widetilde{G}(t,x,ye^{i\alpha})+\frac{\partial}{\partial x}\widetilde{G}(t,x,ye^{i\alpha})\Big|\big|e^{ia(t)(ye^{i\alpha}-x)^2}F(ye^{i\alpha})\big| \notag \\
&\hspace{1.5cm}\leq\Big(2|a(t)||x-ye^{i\alpha}|\widetilde{A}_0(t,x)+\widetilde{A}_1(t,x)\Big)e^{-a(t)\sin(2\alpha)y^2}e^{\widetilde{B}_0(t,x)|y|^{\widetilde{p}}+2a(t)\sin(\alpha)|xy|}. \label{Eq_GtildexF_estimate}
\end{align}
Since $\widetilde{A}_0$, $\widetilde{A}_1$, $\widetilde{B}_0$, $\widetilde{B}_1$, $a$ are assumed to be continuous, the right hand side of this estimate is integrable on $[x_0,x_1]$. Additionally, the factor $e^{-a(t)\sin(2\alpha)y^2}$ implies integrability with respect to $y\in\mathbb{R}$. Hence we observe absolute integrability on $[x_0,x_1]\times\mathbb{R}$ and the order of integration in \eqref{Eq_Psix} can be  interchanged by the Fubini theorem, i.e.
\begin{equation*}
\Psi(t,x_1)=\Psi(t,x_0)+e^{i\alpha}\int_{x_0}^{x_1}\int_\mathbb{R}\frac{\partial}{\partial x}G(t,x,ye^{i\alpha})F(ye^{i\alpha})dydx.
\end{equation*}
In particular, this shows $\Psi(t,\,\cdot\,)\big|_{(0,\infty)}\in\AC((0,\infty))$, the $x$-derivative exists almost everywhere and is given by
\begin{equation}\label{Eq_Psix_Fresnel}
\frac{\partial}{\partial x}\Psi(t,x)=e^{i\alpha}\int_\mathbb{R}\frac{\partial}{\partial x}G(t,x,ye^{i\alpha})F(ye^{i\alpha})dy.
\end{equation}
The same is obviously true for $x<0$ and we conclude $\Psi(t,\,\cdot\,)\in\AC(\dot{\mathbb{R}})$. Using the same argument, also $\frac{\partial}{\partial x}\Psi(t,\,\cdot\,)\in\AC(\dot{\mathbb{R}})$ and $\Psi(\,\cdot\,,t)\in\AC((0,T))$, with second spatial derivative and time derivative almost everywhere given by
\begin{align*}
\frac{\partial^2}{\partial x^2}\Psi(t,x)&=e^{i\alpha}\int_\mathbb{R}\frac{\partial^2}{\partial x^2}G(t,x,ye^{i\alpha})F(ye^{i\alpha})dy, \\
\frac{\partial}{\partial t}\Psi(t,x)&=e^{i\alpha}\int_\mathbb{R}\frac{\partial}{\partial t}G(t,x,ye^{i\alpha})F(ye^{i\alpha})dy.
\end{align*}
This means $\Psi\in\AC_{1,2}((0,T)\times\dot{\mathbb{R}})$ and from \eqref{Eq_G_Schroedinger} we conclude, that the Schrödinger equation \eqref{Eq_Psi_Schroedinger} is satisfied for almost every $t\in(0,T)$, $x\in\dot{\mathbb{R}}$.

\medskip

{\it Step 3.} Next we verify the transmission condition \eqref{Eq_Psi_transmission}. If $M=N=0$, there is nothing to do. In the case $M\neq 0$ and $N=0$, we can estimate the integrand of the integral \eqref{Eq_Psi_Fresnel} as
\begin{equation*}
\big|G(t,x,ye^{i\alpha})F(ye^{i\alpha})\big|\leq\widetilde{A}_0(t,x)e^{-a(t)\sin(2\alpha)y^2+2a(t)\sin(\alpha)|xy|+\widetilde{B}_0(t,x)|y|^{\widetilde{p}}},
\end{equation*}
using the decomposition \eqref{Eq_G_decomposition} as well as the estimate \eqref{Eq_Estimate_Integrand}. Since by Assumption \ref{ass_Greensfunction} (iii) and by \eqref{Eq_Tilde_coefficients} the coefficients $\widetilde{A}_0(t,\,\cdot\,)$ and $\widetilde{B}_0(t,\,\cdot\,)$ are bounded as $x\rightarrow 0^\pm$, the right hand side can be replaced by some integrable and $x$-independent majorant, at least in a neighborhood of $x=0$. Hence we can apply the dominated convergence theorem in \eqref{Eq_Psi_Fresnel}, to get the boundary value
\begin{equation}\label{Eq_Psi_boundary_Fresnel}
\Psi(t,0^\pm)=e^{i\alpha}\int_\mathbb{R}G(t,0^\pm,ye^{i\alpha})F(ye^{i\alpha})dy,\qquad t\in(0,T).
\end{equation}
Moreover, the estimate \eqref{Eq_Estimate_Integrand} in the limit $x\rightarrow 0^\pm$ shows that
\begin{equation*}
|\widetilde{G}(t,0^\pm,z)F(z)|\leq\widetilde{A}_0(t)e^{\widetilde{B}_0(t)|z|^{\widetilde{p}}},\qquad t\in(0,T),\,z\in S_\alpha,
\end{equation*}
for upper bounds $\widetilde{A}_0(t)$ and $\widetilde{B}_0(t)$ of $\widetilde{A}_0(t,\,\cdot\,)$ and $\widetilde{B}_0(t,\,\cdot\,)$ in the limit $x\rightarrow 0^\pm$. Due to Corollary~\ref{cor_Fresnel_integral} it is now possible to write the integral \eqref{Eq_Psi_boundary_Fresnel} again in the real valued form
\begin{equation}\label{Eq_Psi_boundary}
\Psi(t,0^\pm)=\lim\limits_{\varepsilon\rightarrow 0^+}\int_\mathbb{R}e^{-\varepsilon y^2}G(t,0^\pm,y)F(y)dy,\qquad t\in(0,T).
\end{equation}
Since the Green's function satisfies the transmission condition \eqref{Eq_G_transmission} with $N=0$, the same equation carries over to the wave function $\Psi$ and we end up with the stated \eqref{Eq_Psi_transmission}.

\medskip

In the situation $N\neq 0$, also the coefficients $\widetilde{A}_1(t,\,\cdot\,)$ and $\widetilde{B}_1(t,\,\cdot\,)$ from the estimate \eqref{Eq_GtildexF_estimate} are bounded as $x\rightarrow 0^\pm$ by Assumption \ref{ass_Greensfunction} (iii). Hence we are allowed to carry the limit $x\rightarrow 0^\pm$ inside the integral \eqref{Eq_Psix_Fresnel} and get
\begin{equation*}
\frac{\partial}{\partial x}\Psi(t,0^\pm)=e^{i\alpha}\int_\mathbb{R}\frac{\partial}{\partial x}G(t,0^\pm,ye^{i\alpha})F(ye^{i\alpha})dy.
\end{equation*}
Since moreover the estimate \eqref{Eq_GtildexF_estimate} in the limit $x\rightarrow 0^\pm$ looks like
\begin{equation*}
\Big|\frac{\partial}{\partial x}G(t,0^\pm,ye^{i\alpha})F(ye^{i\alpha})\Big|\leq\Big(2|a(t)||y|\widetilde{A}_0(t)+\widetilde{A}_1(t)\Big)e^{-a(t)\sin(2\alpha)y^2}e^{\widetilde{B}_0(t)|y|^{\widetilde{p}}},
\end{equation*}
for upper bounds $\widetilde{A}_1(t)$ and $\widetilde{B}_1(t)$ of $\widetilde{A}_1(t,\,\cdot\,)$ and $\widetilde{B}_1(t,\,\cdot\,)$ in the limit $x\rightarrow 0^\pm$. This estimate now allows to apply Corollary \ref{cor_Fresnel_integral} which transforms the integral back onto the real line
\begin{equation}\label{Eq_Psix_boundary}
\frac{\partial}{\partial x}\Psi(t,0^\pm)=\lim\limits_{\varepsilon\rightarrow 0^+}\int_\mathbb{R}e^{-\varepsilon y^2}\frac{\partial}{\partial x}G(t,0^\pm,y)F(y)dy,\qquad t\in(0,T).
\end{equation}
Since we already know by assumption that $G$ satisfies the transmission condition \eqref{Eq_G_transmission}, the integral representations \eqref{Eq_Psi_boundary} and \eqref{Eq_Psix_boundary} show, that it carries over to $\Psi$ and gives \eqref{Eq_Psi_transmission}.

\medskip

\textit{Step 4.} In the next step we verify the initial condition \eqref{Eq_Psi_initial}. To do so, we fix $x\in\dot{\mathbb{R}}$ and with $x_0>|x|$ from Assumption \ref{ass_Greensfunction} (ii) we split up the integral \eqref{Eq_Psi} as
\begin{align}
\Psi(t,x)&=\underbrace{\lim\limits_{\varepsilon\rightarrow 0^+}\int_{-\infty}^{-x_0}e^{-\varepsilon y^2}G(t,x,y)F(y)dy}_{\eqqcolon\Psi_1(t,x)}+\underbrace{\lim\limits_{\varepsilon\rightarrow 0^+}\int_{-x_0}^{x_0}e^{-\varepsilon y^2}G(t,x,y)F(y)dy}_{\eqqcolon\Psi_2(t,x)} \notag \\
&\quad+\underbrace{\lim\limits_{\varepsilon\rightarrow 0^+}\int_{x_0}^\infty e^{-\varepsilon y^2}G(t,x,y)F(y)dy}_{\eqqcolon\Psi_3(t,x)}. \label{Eq_Psi_decomposition}
\end{align}
We will now derive the initial values of this three integrals separately. Starting with $\Psi_3$, it follows from \eqref{Eq_Estimate_Integrand} that the shifted integrand admits the estimate
\begin{equation}\label{Eq_Estimate_shifted_integrand}
|\widetilde{G}(t,x,x_0+z)F(x_0+z)|\leq\widetilde{A}_0(t,x)e^{\widetilde{B}_0(t,x)|x_0+z|^{\widetilde{p}}},\qquad z\in S_\alpha^+,
\end{equation}
and by Proposition \ref{prop_Fresnel_integral} we can write $\Psi_1$ as the Fresnel integral
\begin{align}
\Psi_3(t,x)&=\lim\limits_{\varepsilon\rightarrow 0^+}\int_0^\infty e^{-\varepsilon(x_0+y)^2}G(t,x,x_0+y)F(x_0+y)dy \notag \\
&=e^{i\alpha}\int_0^\infty G(t,x,x_0+ye^{i\alpha})F(x_0+ye^{i\alpha})dy. \label{Eq_Psi3_Fresnel}
\end{align}
Using once more \eqref{Eq_Estimate_shifted_integrand}, this integral can be estimated as
\begin{align}
|\Psi_3(t,x)|&\leq\widetilde{A}_0(t,x)\int_0^\infty e^{-a(t)\sin(2\alpha)y^2-2a(t)\sin(\alpha)(x_0-x)y+\widetilde{B}_0(t,x)(x_0+y)^{\widetilde{p}}}dy \notag \\
&=\frac{\widetilde{A}_0(t,x)}{\sqrt{a(t)}}\int_0^\infty e^{-\sin(2\alpha)y^2-2\sqrt{a(t)}\,\sin(\alpha)(x_0-x)y+\widetilde{B}_0(t,x)\big(x_0+\frac{y}{\sqrt{a(t)}}\big)^{\widetilde{p}}}dy. \label{Eq_Psi3_estimate}
\end{align}
According to \eqref{Eq_Coefficient_extension} and \eqref{Eq_Tilde_coefficients} we know that $\frac{\widetilde{A}_0}{\sqrt{a}}$ and $\widetilde{B}_0$ remain finite in the limit $t\rightarrow 0^+$, and also that $\lim_{t\rightarrow 0^+}a(t)=\infty$. Therefore, since $x_0>x$, the integrand vanishes in the limit $t\rightarrow 0^+$ and so does the whole function
\begin{equation}\label{Eq_Psi3_initial}
\lim\limits_{t\rightarrow 0^+}\Psi_3(t,x)=0.
\end{equation}
For the same reason also
\begin{equation}\label{Eq_Psi1_initial}
\lim\limits_{t\rightarrow 0^+}\Psi_1(t,x)=0.
\end{equation}
For the function $\Psi_2(t,x)$ we first note, that due to the dominated convergence theorem we are allowed to carry the limit $\varepsilon\rightarrow 0^+$ inside the integral and get
\begin{equation}\label{Eq_Psi2_integral}
\Psi_2(t,x)=\int_{-x_0}^{x_0}G(t,x,y)F(y)dy.
\end{equation}
Since $F\in\mathcal{A}_q(\mathbb{C})$ is an entire function, it is in particular $F\in\mathcal{C}^\infty(\mathbb{R})$ and the initial value
\begin{equation}\label{Eq_Psi2_initial}
\lim\limits_{t\rightarrow 0^+}\Psi_2(t,x)=F(x)
\end{equation}
follows from assumption \eqref{Eq_G_initial}. Combining now \eqref{Eq_Psi3_initial}, \eqref{Eq_Psi1_initial} and \eqref{Eq_Psi2_initial} gives the initial value \eqref{Eq_Psi_initial} of $\Psi(t,x)$ and hence finishes Step 4 of the proof.

\medskip

\textit{Step 5.} It is left to check the continuous dependency \eqref{Eq_Psi_convergence} of the wave function on the initial condition. According to the $\mathcal{A}_q$-convergence \eqref{Eq_Ap_convergence} of the initial condition, we define the coefficients
\begin{equation*}
A_n\coloneqq\sup\limits_{z\in\mathbb{C}}|F(z)-F_n(z)|e^{-B|z|^q},\qquad n\in\mathbb{N},
\end{equation*}
for which obviously
\begin{equation}\label{Eq_An_convergence}
\lim\limits_{n\rightarrow\infty}A_n=0\qquad\text{and}\qquad|F(z)-F_n(z)|\leq A_ne^{B|z|^q},\qquad z\in\mathbb{C},
\end{equation}
holds true. Let $K\subseteq\dot{\mathbb{R}}$ be compact, $x_0>0$ such that $K\subseteq[-x_0,x_0]$ and for every $t\in(0,T)$, $x\in K$ we split up the $\Psi$-integral as in \eqref{Eq_Psi_decomposition}. We will now prove the continuous dependency \eqref{Eq_Psi_convergence} for the three parts of the wave function separately. Rewriting the $\Psi_3$-integral as in \eqref{Eq_Psi3_Fresnel} and estimate the difference as in \eqref{Eq_Psi3_estimate} gives
\begin{align*}
|\Psi_3(t,x;F)-\Psi_3(t,x;F_n)|&=\Big|e^{i\alpha}\int_0^\infty G(t,x,x_0+ye^{i\alpha})\big(F(x_0+ye^{i\alpha})-F_n(x_0+ye^{i\alpha})\big)dy\Big| \\
&\leq\frac{\widetilde{A}_n(t,x)}{\sqrt{a(t)}}\int_0^\infty e^{-\sin(2\alpha)y^2-2\sqrt{a(t)}(x_0-x)\sin(\alpha)y+\widetilde{B}_0(t,x)\big(x_0+\frac{y}{\sqrt{a(t)}}\big)^{\widetilde{p}}}dy,
\end{align*}
using the similar coefficient $\widetilde{A}_n(t,x)\coloneqq A_nA_0(t,x)e^{(B+B_0(t,x))2^{\max\{p,q\}}}$ as in \eqref{Eq_Tilde_coefficients}. Since $A_n\overset{n\rightarrow\infty}{\longrightarrow}0$ by \eqref{Eq_An_convergence} this estimate proves
\begin{equation*}
\lim\limits_{n\rightarrow\infty}\Psi_3(t,x;F_n)=\Psi_3(t,x;F).
\end{equation*}
Since moreover $A_0(t,\,\cdot\,)$ and $B_0(t,\,\cdot\,)$ are continuous on $\dot{\mathbb{R}}$ by Assumption \ref{ass_Greensfunction} (iii), this convergence is uniform with respect to $x\in K$. Following the same arguments, one also obtains the convergence
\begin{equation*}
\lim\limits_{n\rightarrow\infty}\Psi_1(t,x;F_n)=\Psi_1(t,x;F).
\end{equation*}
Finally, carrying the limit $\varepsilon\rightarrow 0^+$ inside the $\Psi_2$-integral, as in \eqref{Eq_Psi2_integral}, we can estimate the difference
\begin{align*}
|\Psi_2(t,x;F)-\Psi_2(t,x;F_n)|&=\Big|\int_{-x_0}^{x_0}G(t,x,y)(F(y)-F_n(y))dy\Big| \\
&\leq A_nA_0(t,x)\int_{-x_0}^{x_0}e^{B_0(t,x)|y|^p+B|y|^q}dy \\
&\leq 2x_0A_nA_0(t,x)e^{B_0(t,x)x_0^p+Bx_0^q}.
\end{align*}
Also here, $A_n\overset{n\rightarrow\infty}{\longrightarrow}0$ implies
\begin{equation*}
\lim\limits_{n\rightarrow\infty}\Psi_2(t,x;F_n)=\Psi_2(t,x;F),
\end{equation*}
uniform with respect to $x\in K$. This verifies the convergence \eqref{Eq_Psi_convergence} and finishes the proof.
\end{proof}

Next we prove Corollary \ref{cor_Greensfunction}, where we use the simplified initial condition \eqref{Eq_G_initial_easy} instead of \eqref{Eq_G_initial}.

\begin{proof}[Proof of Corollary \ref{cor_Greensfunction}]
The fact, that the wave function \eqref{Eq_Psi} exists and $\Psi\in\AC_{1,2}((0,T)\times\dot{\mathbb{R}})$ is a solution of \eqref{Eq_Psi_Schroedinger} and \eqref{Eq_Psi_transmission} is the same as in the proof of Theorem \ref{satz_Greensfunction}. Also the continuous dependency result \eqref{Eq_Psi_convergence} can be proven in the same way.

\medskip

The only thing to check is the initial condition \eqref{Eq_Psi_initial}. Since the calculation is principally the same for $x<0$, we only consider $x>0$ here. First of all, we generalize \eqref{Eq_G_initial_easy} in the sense that for any $z(t)\in D_{\alpha,h}$ with $\lim\limits_{t\rightarrow 0^+}z(t)=x$ we have
\begin{equation}\label{Eq_Gtilde_initial_z}
\lim\limits_{t\rightarrow 0^+}\frac{\widetilde{G}(t,x,z(t))}{\sqrt{a(t)}}=\frac{1}{\sqrt{i\pi}}.
\end{equation}
Consider an open ball $B_r(x)$ with radius $0<r<\min\{h,x\sin(\alpha)\}$ around $x$. Then this ball is obviously contained in the interior $\interior(D_{\alpha,h})$ and we are allowed to apply the Cauchy integral formula to write
\begin{align*}
\widetilde{G}(t,x,z(t))-\widetilde{G}(t,x,x)&=\frac{1}{2\pi i}\int_{|z-x|=r}\Big(\frac{\widetilde{G}(t,x,z)}{z-z(t)}-\frac{\widetilde{G}(t,x,z)}{z-x}\Big)dz \\
&=\frac{z(t)-x}{2\pi i}\int_{|z-x|=r}\frac{\widetilde{G}(t,x,z)}{(z-z(t))(z-x)}dz \\
&=\frac{z(t)-x}{2\pi}\int_0^{2\pi}\frac{\widetilde{G}(t,x,x+re^{i\theta})}{x+re^{i\theta}-z(t)}d\theta.
\end{align*}
Using \eqref{Eq_Gtilde_estimate}, we can estimate the integrand to get
\begin{align*}
\big|\widetilde{G}(t,x,z(t))-\widetilde{G}(t,x,x)\big|&\leq\frac{A_0(t,x)|z(t)-x|}{2\pi}\int_0^{2\pi}\frac{e^{B_0(t,x)|x+re^{i\theta}|^p}}{|x+re^{i\theta}-z(t)|}d\theta \\
&\leq\frac{A_0(t,x)|z(t)-x|}{r-|z(t)-x|}e^{B_0(t,x)(|x|+r)^p}.
\end{align*}
Since $\frac{A_0(t,x)}{\sqrt{a(t)}}$ and $B_0(t,x)$ are bounded as $t\rightarrow 0^+$ and $\lim_{t\rightarrow 0^+}z(t)=x$, it follows, that
\begin{equation*}
\lim\limits_{t\rightarrow 0^+}\frac{|\widetilde{G}(t,x,z(t))-\widetilde{G}(t,x,x)|}{\sqrt{a(t)}}=0.
\end{equation*}
With \eqref{Eq_G_initial_easy} and the decomposition \eqref{Eq_G_decomposition}, we then obtain the limit \eqref{Eq_Gtilde_initial_z}, namely
\begin{equation*}
\lim\limits_{t\rightarrow 0^+}\frac{\widetilde{G}(t,x,z(t))}{\sqrt{a(t)}}=\lim\limits_{t\rightarrow 0^+}\frac{\widetilde{G}(t,x,x)}{\sqrt{a(t)}}=\lim\limits_{t\rightarrow 0^+}\frac{G(t,x,x)}{\sqrt{a(t)}}=\frac{1}{\sqrt{i\pi}}.
\end{equation*}
Next, we split up the integral \eqref{Eq_Psi} as
\begin{align*}
\Psi(t,x)&=e^{i\alpha}\int_{-\infty}^0G(t,x,ye^{i\alpha})F(ye^{i\alpha})dy+\int_0^xG(t,x,y)F(y)dy \\
&\hspace{3cm}+e^{i\alpha}\int_0^\infty G(t,x,x+ye^{i\alpha})F(x+ye^{i\alpha})dy,
\end{align*}
which is the same decomposition as in \eqref{Eq_Psi_decomposition}, with the interval $[-x_0,x_0]$ replaced by $[0,x]$, and the subsequent rewriting of the integrals \eqref{Eq_Psi3_Fresnel} and \eqref{Eq_Psi2_integral}. Using the Cauchy theorem, we change the integration path $0\rightarrow x$ once more to $0\rightarrow x-\delta e^{i\alpha}\rightarrow x$ where $\delta>0$ is small enough such that $x-\delta e^{i\alpha}\in\interior(D_{\alpha,h})$. I.e. we end up with the representation
\begin{equation}\label{Eq_Psi_Fresnel_decomposition}
\Psi(t,x)=\int_{\gamma_1}G(t,x,z)F(z)dz+\int_{\gamma_2}G(t,x,z)F(z)dz+\int_{\gamma_3}G(t,x,z)F(z)dz,
\end{equation}
using the three complex paths

\begin{minipage}{0.44\textwidth}
\begin{align*}
\gamma_1&\coloneqq\Set{ye^{i\alpha} | y\leq 0}, \\
\gamma_2&\coloneqq\Set{s(x-\delta e^{i\alpha}) | 0\leq s\leq 1}, \\
\gamma_3&\coloneqq\Set{x+ye^{i\alpha} | y\geq-\delta}.
\end{align*}
\end{minipage}
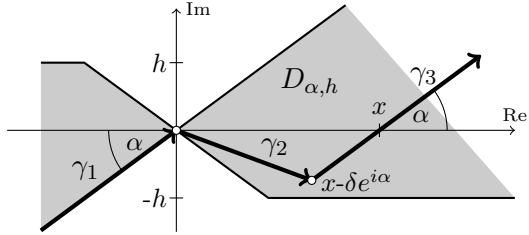
\begin{minipage}{0.55\textwidth}
\begin{tikzpicture}[scale=0.9]
\fill[black!20] (5,-1)--(1.36,-1)--(0,0)--(2.5,1.85);
\fill[black!20] (-2,1)--(-1.36,1)--(0,0)--(-2,-1.48);
\draw[ultra thick,->] (-2,-1.48)--(0,0);
\draw[ultra thick,->] (0,0)--(2,-0.74);
\draw[ultra thick,->] (2,-0.74)--(4.5,1.11);
\draw[->] (-2.5,0)--(5,0) node[anchor=south] {\scriptsize{$\Re$}};
\draw[->] (0,-1.5)--(0,1.8) node[anchor=west] {\scriptsize{$\Im$}};
\draw[thick] (-2,-1.48)--(2.5,1.85);
\draw[thick] (-2,1)--(-1.36,1)--(1.36,-1)--(5,-1);
\draw (4,0) arc (0:36.86:1) (3.35,0.2) node[anchor=west] {$\alpha$};
\draw (-1,0) arc (180:216.86:1) (-0.35,-0.2) node[anchor=east] {$\alpha$};
\draw (-0.1,1)--(0.1,1) (0,1) node[anchor=east] {$h$};
\draw (-0.1,-1)--(0.1,-1) (0,-1) node[anchor=east] {-$h$};
\draw (1.4,0.7) node[anchor=west] {\large{$D_{\alpha,h}$}};
\draw (3,-0.1)--(3,0.1) node[anchor=south] {$x$};
\draw (-1,-0.6) node[anchor=east] {\large{$\gamma_1$}};
\draw (1.8,-0.25) node[anchor=east] {\large{$\gamma_2$}};
\draw (4,0.8) node[anchor=east] {\large{$\gamma_3$}};
\draw[fill=white] (0,0) circle (0.06);
\draw[fill=white] (2,-0.74) circle (0.06) node[anchor=west] {$x$-$\delta e^{i\alpha}$};
\end{tikzpicture}
\end{minipage}

\medskip

We will now consider all three integrals separately. Starting with the integral along $\gamma_1$, the estimate \eqref{Eq_Estimate_Integrand} with the coefficients \eqref{Eq_Tilde_coefficients} yields the estimate
\begin{align*}
\Big|\int_{\gamma_1}G(t,x,z)F(z)dz\Big|&=\Big|e^{i\alpha}\int_{-\infty}^0e^{ia(t)(ye^{i\alpha}-x)^2}\widetilde{G}(t,x,ye^{i\alpha})F(ye^{i\alpha})dy\Big| \\
&\leq\widetilde{A}_0(t,x)\int_{-\infty}^0e^{-a(t)\sin(2\alpha)y^2+2a(t)\sin(\alpha)xy+\widetilde{B}_0(t,x)|y|^{\widetilde{p}}}dy \\
&=\frac{\widetilde{A}_0(t,x)}{\sqrt{a(t)}}\int_{-\infty}^0e^{-\sin(2\alpha)y^2-2\sqrt{a(t)}\sin(\alpha)x|y|+\widetilde{B}_0(t,x)\big(\frac{|y|}{\sqrt{a(t)}}\big)^{\widetilde{p}}}dy.
\end{align*}
According to \eqref{Eq_Coefficient_extension} and \eqref{Eq_Tilde_coefficients}, we know that $\frac{\widetilde{A}_0}{\sqrt{a}}$ and $\widetilde{B}_0$ remain finite in the limit $t\rightarrow 0^+$, and also that $\lim_{t\rightarrow 0^+}a(t)=\infty$. Therefore, since $x>0$, the integrand vanishes in the limit $t\rightarrow 0^+$ and so does the whole integral
\begin{equation}\label{Eq_gamma1_initial}
\lim\limits_{t\rightarrow 0^+}\int_{\gamma_1}G(t,x,z)F(z)dz=0.
\end{equation}
For the integral along $\gamma_2$ we use the estimate \eqref{Eq_Estimate_Integrand} to get
\begin{align*}
\Big|\int_{\gamma_2}G(t,x,z)F(z)dz\Big|&=\Big|(x-\delta e^{i\alpha})\int_0^1G\big(t,x,s(x-\delta e^{i\alpha})\big)F\big(s(x-\delta e^{i\alpha})\big)ds\Big| \\
&\leq\widetilde{A}_0(t,x)|x-\delta e^{i\alpha}|\int_0^1e^{-2\delta a(t)\sin(\alpha)s\big(x-s\Re(x-\delta e^{i\alpha})\big)+\widetilde{B}_0(t,x)|s(x-\delta e^{i\alpha})|^{\widetilde{p}}}ds \\
&\leq\widetilde{A}_0(t,x)|x-\delta e^{i\alpha}|e^{\widetilde{B}_0(t,x)|x-\delta e^{i\alpha}|^{\widetilde{p}}}\int_0^\infty e^{-2\delta a(t)\sin(\alpha)s\big(x-\Re(x-\delta e^{i\alpha})\big)}ds \\
&=\frac{\widetilde{A}_0(t,x)|x-\delta e^{i\alpha}|}{\delta^2a(t)\sin(2\alpha)}e^{\widetilde{B}_0(t,x)|x-\delta e^{i\alpha}|^{\widetilde{p}}},
\end{align*}
where we used in the estimate of the third line that $\Re(x-\delta e^{i\alpha})>0$. Since again $\frac{\widetilde{A}_0}{\sqrt{a}}$, $\widetilde{B}_0$ are bounded as $t\rightarrow 0^+$ and $\lim_{t\rightarrow 0^+}a(t)=\infty$, also
\begin{equation}\label{Eq_gamma2_initial}
\lim\limits_{t\rightarrow 0^+}\int_{\gamma_2}G(t,x,z)F(z)dz=0.
\end{equation}
Finally, the integral along $\gamma_3$ can be parametrized as
\begin{align*}
\int_{\gamma_3}G(t,x,z)F(z)dz&=e^{i\alpha}\int_{-\delta}^\infty G(t,x,x+ye^{i\alpha})F(x+ye^{i\alpha})dy \\
&=\frac{e^{i\alpha}}{\sqrt{a(t)}}\int_{-\delta\sqrt{a(t)}}^\infty G\Big(t,x,x+\frac{ye^{i\alpha}}{\sqrt{a(t)}}\Big)F\Big(x+\frac{ye^{i\alpha}}{\sqrt{a(t)}}\Big)dy.
\end{align*}
Again by \eqref{Eq_Estimate_Integrand}, the integrand can be estimated as
\begin{equation*}
\Big|G\Big(t,x,x+\frac{ye^{i\alpha}}{\sqrt{a(t)}}\Big)F\Big(x+\frac{ye^{i\alpha}}{\sqrt{a(t)}}\Big)\Big|\leq\frac{\widetilde{A}_0(t,x)}{\sqrt{a(t)}}e^{-y^2\sin(2\alpha)+\widetilde{B}_0(t,x)\big|x+\frac{ye^{i\alpha}}{\sqrt{a(t)}}\big|^{\widetilde{p}}}.
\end{equation*}
Once more, $\frac{\widetilde{A}_0}{\sqrt{a}}$, $\widetilde{B}_0$ are bounded as $t\rightarrow 0^+$ and $\lim_{t\rightarrow 0^+}a(t)=\infty$, and hence this upper bound can be made $t$-independent in a neighborhood of $t=0^+$. By the dominated convergence theorem we are then allowed to carry the limit $t\rightarrow 0^+$ inside the integral, and using also \eqref{Eq_Gtilde_initial_z} the $\gamma_3$-integral then becomes
\begin{align}
\lim\limits_{t\rightarrow 0^+}\int_{\gamma_3}G(t,x,z)F(z)dz&=e^{i\alpha}\int_\mathbb{R}\lim\limits_{t\rightarrow 0^+}\frac{1}{\sqrt{a(t)}}G\Big(t,x,x+\frac{ye^{i\alpha}}{\sqrt{a(t)}}\Big)F\Big(x+\frac{ye^{i\alpha}}{\sqrt{a(t)}}\Big)dy \notag \\
&=\frac{e^{i\alpha}}{\sqrt{i\pi}}\int_\mathbb{R}e^{iy^2e^{2i\alpha}}F(x)dy=F(x). \label{Eq_gamma3_initial}
\end{align}
Combining now \eqref{Eq_gamma1_initial}, \eqref{Eq_gamma2_initial} and \eqref{Eq_gamma3_initial} gives the initial value \eqref{Eq_Psi_initial}.
\end{proof}

\section{Stability of superoscillations and supershifts}\label{sec_Stability_of_superoscillations_and_supershifts}

It is a question almost as old as superoscillations itself: What happens to a superoscillating function as it evolves in time, when interacting with some quantum mechanical system? In other words: What happens if we put some superoscillatory sequence $(F_n)_n$ as initial condition of the time dependent Schrödinger equation \eqref{Eq_Psi_Cauchy}? Will the sequence of solutions $(\Psi(t,x;F_n))_n$ still be superoscillating at times $t>0$?

\medskip

Motivated by the example \eqref{Eq_Example}, the two defining properties of superoscillations are on the one hand the decomposition into plane waves
\begin{equation}\label{Eq_Fn}
F_n(z)=\sum\limits_{j=0}^nC_j(n)e^{ik_j(n)z},\qquad z\in\mathbb{C},
\end{equation}
with coefficients $C_j(n)\in\mathbb{C}$ and frequencies $k_j(n)\in[-k_0,k_0]$, for some $k_0>0$. On the other hand the functions converge as
\begin{equation}\label{Eq_Fn_convergence}
\lim\limits_{n\rightarrow\infty}F_n(z)=e^{ikz},\qquad\text{in }\mathcal{A}_1(\mathbb{C}),
\end{equation}
to some plane wave with frequency $k\in\mathbb{R}\setminus[-k_0,k_0]$. The continuous dependency result of Theorem~\ref{satz_Greensfunction} now already indicates some convergence
\begin{equation}\label{Eq_Psin_convergence}
\lim\limits_{n\rightarrow\infty}\Psi(t,x;F_n)=\Psi(t,x;e^{ik\,\cdot\,}).
\end{equation}
However, although the limit function $\Psi(t,x;e^{ik\,\cdot\,})$ may admit some oscillatory behaviour, it is by no means expected (and also not true) that it is again a plane wave $e^{ik(t)x}$. Also, since Theorem \ref{satz_Greensfunction} only gives uniform convergence on compact subsets of $\dot{\mathbb{R}}$, the desired $\mathcal{A}_1$-convergene in the variable $x$ may in general fail. Already the extension to a holomorphic function will not be possible in general, since it is intuitively clear, that some potential, having for example discontinuities, leads to a wave function with is no longer holomorphic although the initial value was.

\medskip

To overcome this dilemma, the notion of \textit{supershift} was introduced in \cite{CSSY19_1} and appeared in slightly different variations in subsequent publications. In the following we will stick to the one in \cite{ABCS21_2}, namely motivated by \eqref{Eq_Fn} and \eqref{Eq_Fn_convergence} we define:

\medskip

\begin{defi}[Supershift]\label{defi_Supershift}
Let $\mathcal{O},\mathcal U\subseteq\mathbb{C}$ with $\mathcal{U}\subsetneqq\mathcal{O}$ and $X$ be a metric space. Consider
\begin{equation}\label{Eq_varphi_kappa}
\varphi_\kappa:X\rightarrow\mathbb C, \qquad \kappa\in\mathcal O,
\end{equation}
a family of complex valued functions. We say that a sequence of functions $(\Phi_n)_n$ of the form
\begin{equation}\label{Eq_Supershift_function}
\Phi_n(s)=\sum\limits_{l=0}^nC_l(n)\varphi_{\kappa_l(n)}(s)
,\qquad s\in X,
\end{equation}
with coefficients $C_l(n)\in\mathbb{C}$, $\kappa_l(n)\in\mathcal U$, admits a \emph{supershift}, if there exists $\kappa\in\mathcal{O}\setminus\mathcal U$, such that
\begin{equation}\label{Eq_Supershift_convergence}
\lim\limits_{n\rightarrow\infty}\Phi_n(s)=\varphi_\kappa(s),\qquad s\in X,
\end{equation}
converges uniformly on compact subsets of $X$.
\end{defi}

\medskip

This definition of supershift mainly means, that we forget about the oscillatory behaviour of the plane waves $e^{i\kappa z}$ and replace them by some arbitrary functions $\varphi_\kappa(z)$. Also the $\mathcal{A}_1$-convergence \eqref{Eq_Fn_convergence} gets weakened to the uniform convergence on compact sets \eqref{Eq_Supershift_convergence}. If the used function $\varphi_\kappa$ admit any kind of oscillatory behaviour may be investigated for each potential independently.

\medskip

\begin{bem}
Since we introduced the supershift property with the aim to generalize superoscillations, it is obligatory that the example functions \eqref{Eq_Example} fit into Definition~\ref{defi_Supershift}. But this is indeed the case, choosing $X=\mathbb{C}$, $\mathcal{O}=\mathbb{R}$, $\mathcal{U}=[-1,1]$ and $\varphi_\kappa(z)=e^{i\kappa z}$, see also \cite[Example 4.3]{ABCS21_2}.
\end{bem}

\medskip

The first main result of this section is Theorem \ref{satz_Supershift_property} on the supershift property of the solution of the Schrödinger equation \eqref{Eq_Psi_Cauchy}, which can be viewed as a corollary of the continuous dependence result of Theorem \ref{satz_Greensfunction}. Roughly speaking, we consider a family of initial conditions that admit a supershift (with respect to a slightly stronger form of convergence as in Definition \ref{defi_Supershift}) and conclude that the corresponding solutions of the Schrödinger equation admit a similar type of supershift, see also Remark \ref{bem_Supershift}.

\medskip

\begin{satz}\label{satz_Supershift_property}
Let the function $G$ be as in Assumption \ref{ass_Greensfunction} and $\mathcal{O},\mathcal{U}\subseteq\mathbb{C}$ with $\mathcal{U}\subsetneqq\mathcal{O}$. For some $q\in(0,2)$ consider a family of functions $\varphi_\kappa\in\mathcal{A}_q(\mathbb{C})$, $\kappa\in\mathcal{O}$ and a sequence
\begin{equation}\label{Eq_Initial_Fn}
F_n(z)=\sum\limits_{l=0}^nC_l(n)\varphi_{\kappa_l(n)}(z),\qquad z\in\mathbb{C},
\end{equation}
of initial conditions with coefficients $C_l(n)\in\mathbb{C}$, $\kappa_l(n)\in\mathcal U$. If now
\begin{equation}\label{Eq_phi_kappa_convergence}
\lim\limits_{n\rightarrow\infty}F_n=\varphi_\kappa\quad\text{in }\mathcal{A}_q(\mathbb{C}),
\end{equation}
for some $\varphi_\kappa$ with $\kappa\in\mathcal{O}\setminus\mathcal{U}$, then the solutions of the Cauchy problem \eqref{Eq_Psi_Cauchy} converges as
\begin{equation}\label{Eq_Psi_kappa_convergence}
\lim\limits_{n\rightarrow\infty}\Psi(t,x;F_n)=\lim\limits_{n\rightarrow\infty}\sum\limits_{l=0}^nC_l(n)\Psi(t,x;\varphi_{\kappa_l(n)})=\Psi(t,x;\varphi_\kappa),
\end{equation}
for every $t\in(0,T)$ and uniformly on compact subsets $\dot{\mathbb{R}}$.
\end{satz}

\medskip

\begin{bem}\label{bem_Supershift}
Since the convergence \eqref{Eq_phi_kappa_convergence} implies uniform convergence on all compact subsets of $\mathbb{R}$, it is clear, that the initial conditions $(F_n)_n$ in \eqref{Eq_Initial_Fn} admit the supershift property of Definition~\ref{defi_Supershift} in the space $X=\dot{\mathbb{R}}$. Furthermore, the convergence \eqref{Eq_Psi_kappa_convergence} shows that for every $t\in(0,T)$ also the solutions $\Psi(t,\,\cdot\;;F_n)$ admit a supershift with respect to the functions $\phi_\kappa(t,x)\coloneqq\Psi(t,x;\varphi_\kappa)$.
\end{bem}

\begin{proof}[Proof of Theorem \ref{satz_Supershift_property}]
The fact, that the convergence \eqref{Eq_phi_kappa_convergence} leads to the convergence \eqref{Eq_Psi_kappa_convergence}, was already subject of Theorem \ref{satz_Greensfunction}. Moreover, splitting the solutions $\Psi(t,x;F_n)$ into the given linear combination is allowed due to the linearity of the Schrödinger equation with respect to the initial condition, i.e.
\begin{equation*}
\Psi(t,x;F_n)=\Psi\Big(t,x;\sum\limits_{l=0}^nC_l(n)\varphi_{\kappa_l(n)}\Big)=\sum\limits_{l=0}^nC_l(n)\Psi(t,x;\varphi_{\kappa_l(n)}). \qedhere
\end{equation*}
\end{proof}

If the sequence $(\Phi_n)_n$ in \eqref{Eq_Supershift_function} admits a supershift, then the values of $\varphi_\kappa$ for some $\kappa\in\mathcal O$, outside the smaller set $\mathcal{U}$, can be calculated by \eqref{Eq_Supershift_convergence} only using values $\varphi_{\kappa_l(n)}$ at the points $\kappa_l(n)$ inside $\mathcal U$. Hence, informally speaking, when considering the mapping $\kappa\mapsto\varphi_\kappa$ in the $\kappa$-variable, this feels like a property of analytic functions. Although we are not able to prove that the supershift property implies analyticity in general, the following Theorem \ref{satz_Analyticity} at least shows, that analyticity in the $\kappa$-variable of the initial condition implies analyticity in the $\kappa$-variable of the wave function.

\medskip

\begin{satz}\label{satz_Analyticity}
Let $G$ be as in Assumption \ref{ass_Greensfunction}. For some open set $\mathcal{O}\subseteq\mathbb{C}$ and $q\in(0,2)$, we consider a family of functions $\varphi_\kappa\in\mathcal{A}_q(\mathbb{C})$, $\kappa\in\mathcal{O}$, such that
\begin{equation}\label{Eq_varphi_estimate}
|\varphi_\kappa(z)|\leq A(\kappa)e^{B(\kappa)|z|^q},\qquad z\in\mathbb{C},
\end{equation}
is satisfied for $A(\kappa),B(\kappa)\geq 0$ continuously depending on $\kappa$. If for every $z\in S_\alpha$ the mapping
\begin{equation*}
\mathcal{O}\ni\kappa\mapsto\varphi_\kappa(z)
\end{equation*}
is holomorphic, then for every fixed $t\in(0,T)$, $x\in\dot{\mathbb{R}}$, the mapping
\begin{equation*}
\mathcal{O}\ni\kappa\mapsto\Psi(t,x;\varphi_\kappa)
\end{equation*}
is holomorphic as well, where $\Psi(t,x;\varphi_\kappa)$ is the solution of the Cauchy problem \eqref{Eq_Psi_Cauchy} with initial value $\varphi_\kappa$.
\end{satz}

\begin{proof}
Fix $t\in(0,T)$, $x\in\dot{\mathbb{R}}$. Then for any closed triangle $\Delta\subseteq\mathcal{O}$, we have the path integral
\begin{equation}\label{Eq_Psi_kappa2}
\int_\Delta\Psi(t,x;\varphi_\kappa)d\kappa=\int_\Delta e^{i\alpha}\int_\mathbb{R}G(t,x,ye^{i\alpha})\varphi_\kappa(ye^{i\alpha})dyd\kappa,
\end{equation}
due to the representation \eqref{Eq_Psi_Fresnel} of the wave function. Here $\alpha\in(0,\frac{\pi}{2})$ is the angle of the double sector $S_\alpha$ in Assumption \ref{ass_Greensfunction}. In order to interchange the order of integration, we have to prove absolute integrability of the double integral. Firstly, the estimate
\begin{equation}\label{Eq_Psi_kappa1}
\big|G(t,x,ye^{i\alpha})\varphi_\kappa(ye^{i\alpha})\big|\leq A(\kappa)A_0(t,x)e^{-a(t)\sin(2\alpha)y^2+2a(t)\sin(\alpha)|xy|+B_0(t,x)|y|^p+B(\kappa)|y|^q},
\end{equation}
follows from \eqref{Eq_Gtilde_estimate} as well as \eqref{Eq_varphi_estimate}, and shows that the $y$-integral is absolutely convergent. Moreover, the coefficients $A(\kappa),B(\kappa)$ are assumed to be continuous and hence this upper bound can be uniformly estimated on the compact triangle $\Delta$. This means, that the right hand side of \eqref{Eq_Psi_kappa1} can be replaced by some $\kappa$-independent and $y$-integrable upper bound. Hence, the double integral \eqref{Eq_Psi_kappa2} is absolutely convergent and we are allowed to interchange the order of integration
\begin{equation*}
\int_\Delta\Psi(t,x;\varphi_\kappa)d\kappa=e^{i\alpha}\int_\mathbb{R}G(t,x,ye^{i\alpha})\int_\Delta\varphi_\kappa(ye^{i\alpha})d\kappa dy.
\end{equation*}
Since the mapping $\mathcal O\ni\kappa\mapsto\varphi_\kappa(ye^{i\alpha})$ is holomorphic the path integral along $\Delta$ vanishes and we get
\begin{equation*}
\int_\Delta\Psi(t,x;\varphi_\kappa)d\kappa=0.
\end{equation*}
Due to the Theorem of Morera this implies the analyticity of $\mathcal{O}\ni\kappa\mapsto\Psi(t,x;\varphi_\kappa)$.
\end{proof}

\section{Examples of Green's functions}\label{sec_Examples_of_Greens_functions}

This section provides two examples of potentials and transmission conditions which are exemplary for the Assumptions \ref{ass_Greensfunction}. The first one is the centrifugal potential $V(t,x)=\frac{\lambda}{x^2}$ in Section \ref{sec_Centrifugal_potential} and the second one is the free particle on $\dot{\mathbb{R}}$, satisfying arbitrary transmission conditions at $x=0^\pm$ in Section \ref{sec_Arbitrary_point_interactions}. In particular the distributional $\delta$- and $\delta'$-potentials as well as boundary conditions of Dirichlet-, Neumann- and Robin-type are special cases of the second example.

\subsection{Centrifugal potential}\label{sec_Centrifugal_potential}

In this subsection we consider the strongly singular centrifugal potential $V(t,x)=\frac{\lambda}{x^2}$ of strength $\lambda\in\mathbb{R}\setminus\{0\}$. For this potential the case $\lambda>0$ was already investigated in \cite{ACST18,CGS19}, and in \cite{ACSS20} they even consider the combined centrifugal and harmonic oscillator potential $V(t,x)=\frac{\lambda}{x^2}+\omega x^2$, which would also be possible in the following considerations, but in order to avoid technical difficulties we omit this discussion here. The case $\lambda<0$ on the other hand was, to the best of our knowledge, not yet treated with respect to stability of superoscillations and is a novelty of this paper.

\medskip

Note that, realized by the step function $\Theta$, the upcoming Green's functions \eqref{Eq_G_centrifugal_Hankel} and \eqref{Eq_G_centrifugal_Bessel} vanish for $xy<0$, which is due to the fact that the $\frac{1}{x^2}$-potential is too singular at $x=0$ to allow any information exchange between the two halflines. Additionally, this non-integrable singularity automatically implies Dirichlet boundary conditions $\Psi(t,0^+)=\Psi(t,0^-)=0$, see for example \cite{EG06} for justification.

\begin{prop}\label{prop_Example_centrifugal_Hankel}
The Green's function of the attractive centrifugal potential $V(t,x)=\frac{\lambda}{x^2}$, $\lambda<0$, is given by
\begin{equation}\label{Eq_G_centrifugal_Hankel}
G(t,x,y)=\frac{\Theta(xy)\sqrt{xy}}{4i^{\nu+1}t}e^{-\frac{x^2+y^2}{4it}}H_\nu^{(2)}\Big(\frac{xy}{2t}\Big),\qquad t>0,\,x,y\in\dot{\mathbb{R}}.
\end{equation}
Here $H_\nu^{(2)}$ is the Hankel function of the second kind and $\nu\coloneqq\sqrt{\frac{1}{4}+\lambda}$ is real valued for $-\frac{1}{4}\leq\lambda<0$, or purely imaginary for $\lambda<-\frac{1}{4}$. With the transmission matrices $M=I$ the identity matrix and $N=0$ the zero matrix, this Green's function satisfies the Assumption \ref{ass_Greensfunction} with $S_\alpha$ replaced by $D_{\alpha,h}$, for any $\alpha\in(0,\frac{\pi}{2})$, $h>0$, and \eqref{Eq_G_initial} replaced by \eqref{Eq_G_initial_easy}.
\end{prop}

\begin{prop}\label{prop_Example_centrifugal_Bessel}
The Green's function of the repulsive centrifugal potential $V(t,x)=\frac{\lambda}{x^2}$, $\lambda>0$ is given by
\begin{equation}\label{Eq_G_centrifugal_Bessel}
G(t,x,y)=\frac{\Theta(xy)\sqrt{xy}}{2i^{\nu+1}t}e^{-\frac{x^2+y^2}{4it}}J_\nu\Big(\frac{xy}{2t}\Big),\qquad t>0,\,x,y\in\dot{\mathbb{R}},
\end{equation}
where $J_\nu$ denotes the Bessel function of the first kind and $\nu\coloneqq\sqrt{\frac{1}{4}+\lambda}$. With the transmission matrices $M=I$ the identity matrix and $N=0$ the zero matrix, this Green's function satisfies the Assumption \ref{ass_Greensfunction}.
\end{prop}

\begin{proof}[Proof of Proposition \ref{prop_Example_centrifugal_Hankel}]
First of all, for every $t>0$, $x\in\dot{\mathbb{R}}$ the function $G(t,x,\,\cdot\,)$ holomorphically extends to $\Set{z\in\mathbb{C} | \Re(z)\neq 0}$ by
\begin{equation*}
G(t,x,z)=\frac{\Theta(\pm x)\sqrt{xz}}{4i^{\nu+1}t}e^{-\frac{x^2+z^2}{4it}}H_\nu^{(2)}\Big(\frac{xz}{2t}\Big),\qquad\pm\Re(z)>0.
\end{equation*}
In particular, $G(t,x,\,\cdot\,)$ is holomorphic on the domain $D_{\alpha,h}$ in \eqref{Eq_Dalphah} for any $\alpha\in(0,\frac{\pi}{2})$, $h>0$. Moreover, the decomposition \eqref{Eq_G_decomposition} is satisfied using $a(t)=\frac{1}{4t}$ and
\begin{equation*}
\widetilde{G}(t,x,z)=\frac{\Theta(\pm x)\sqrt{xz}}{4i^{\nu+1}t}e^{-\frac{xz}{2it}}H_\nu^{(2)}\Big(\frac{xz}{2t}\Big),\qquad\pm\Re(z)>0.
\end{equation*}
It is obvious, that for fixed $z\in\mathbb{C}$ with $\Re(z)\neq 0$ we have $G(\,\cdot\,,\,\cdot\,,z)\in\AC_{1,2}((0,T)\times\dot{\mathbb{R}})$. In order to see, that it is a solution of \eqref{Eq_G_Schroedinger}, we explicitly calculate the derivatives
\begin{align}
\frac{\partial}{\partial x}\widetilde{G}(t,x,z)&=\Big(\frac{1}{2x}-\frac{z}{2it}\Big)\widetilde{G}(t,x,z)+\frac{\Theta(\pm x)z\sqrt{xz}}{8i^{\nu+1}t^2}e^{-\frac{xz}{2it}}H_\nu^{(2)'}\Big(\frac{xz}{2t}\Big), \notag \\
\frac{\partial^2}{\partial x^2}\widetilde{G}(t,x,z)&=\Big(\frac{\nu^2-\frac{1}{4}}{x^2}-\frac{z}{2itx}-\frac{z^2}{2t^2}\Big)\widetilde{G}(t,x,z)+\frac{\Theta(\pm x)z^2\sqrt{xz}}{8i^\nu t^3}e^{-\frac{xz}{2it}}H_\nu^{(2)'}\Big(\frac{xz}{2t}\Big), \label{Eq_Gtilde_centrifugal_derivatives_Hankel} \\
\frac{\partial}{\partial t}\widetilde{G}(t,x,z)&=\Big(\frac{xz}{2it^2}-\frac{1}{t}\Big)\widetilde{G}(t,x,z)-\frac{\Theta(\pm x)(xz)^{\frac{3}{2}}}{8i^{\nu+1}t^3}e^{-\frac{xz}{2it}}H_\nu^{(2)'}\Big(\frac{xz}{2t}\Big), \notag
\end{align}
where for the second spatial derivative we used the Bessel differential equation
\begin{equation*}
w^2H_\nu^{(2)''}(w)+wH_\nu^{(2)'}(w)+(w^2-\nu^2)H_\nu^{(2)}(w)=0,\qquad\Re(w)>0.
\end{equation*}
Using $\lambda=\nu^2-\frac{1}{4}$, these derivatives now satisfy 
\begin{equation*}
i\frac{\partial}{\partial t}\widetilde{G}(t,x,z)=\Big(-\frac{\partial^2}{\partial x^2}+\frac{x-z}{it}\frac{\partial}{\partial x}+\frac{1}{2it}+\frac{\lambda}{x^2}\Big)\widetilde{G}(t,x,z),
\end{equation*}
which is equivalent to \eqref{Eq_G_Schroedinger} using the decomposition \eqref{Eq_G_decomposition}. From \cite[Eq.(9.1.4),(9.1.7),(9.1.8)]{AS72} and the fact that $0\leq\Re(\nu)<\frac{1}{2}$ we get the convergence
\begin{equation}\label{Eq_Hankel_zero_asymptotics}
\sqrt{w}\,H_\nu^{(2)}(w)\rightarrow 0,\qquad\text{as }w\rightarrow 0,\,\Re(w)>0.
\end{equation}
Hence we conclude for every $t>0$, $\pm y>0$ the boundary values \eqref{Eq_G_transmission} by
\begin{equation*}
G(t,x,y)=\frac{\Theta(\pm x)\sqrt{xy}}{4i^{\nu+1}t}e^{-\frac{x^2+y^2}{4it}}H_\nu^{(2)}\Big(\frac{xy}{2t}\Big)\overset{x\rightarrow 0}{\longrightarrow}0.
\end{equation*}
Moreover, with the asymptotics
\begin{equation}\label{Eq_Hankel_inf_asymptotics}
H_\nu^{(2)}(w)=e^{-iw}\bigg(\frac{\sqrt{2}\,i^{\nu+\frac{1}{2}}}{\sqrt{\pi w}}+\mathcal{O}\Big(\frac{1}{w^{\frac{3}{2}}}\Big)\bigg),\qquad\text{as }w\rightarrow\infty,\,\Re(w)>0,
\end{equation}
from \cite[Eq.(9.2.8)]{AS72}, it follows that for every $x\in\dot{\mathbb{R}}$ the initial condition \eqref{Eq_G_initial_easy} is satisfied by
\begin{equation*}
\frac{G(t,x,x)}{\sqrt{a(t)}}=\frac{|x|}{2i^{\nu+1}\sqrt{t}}e^{-\frac{x^2+z^2}{4it}}H_\nu^{(2)}\Big(\frac{x^2}{2t}\Big)\overset{t\rightarrow 0^+}{\longrightarrow}\frac{1}{\sqrt{i\pi}}.
\end{equation*}
Combining the asymptotics \eqref{Eq_Hankel_zero_asymptotics} and \eqref{Eq_Hankel_inf_asymptotics}, there also exists some constant $C_\nu\geq 0$ such that
\begin{equation}\label{Eq_Hankel_estimate}
|H_\nu^{(2)}(w)|\leq\frac{C_\nu}{\sqrt{|w|}}e^{\Im(w)},\qquad\Re(w)>0.
\end{equation}
Hence we can estimate
\begin{equation}\label{Eq_Gtilde_estimate_Hankel}
|\widetilde{G}(t,x,z)|=\frac{\Theta(\pm x)\sqrt{|xz|}}{4t}e^{-\frac{x\Im(z)}{2t}}\Big|H_\nu^{(2)}\Big(\frac{xz}{2t}\Big)\Big|\leq\frac{C_\nu}{2\sqrt{2t}},\qquad\pm\Re(z)>0,
\end{equation}
which is \eqref{Eq_Gtilde_estimate} with the coefficients $A_0(t,x)=\frac{C_\nu}{2\sqrt{2t}}$ and $B_0(t,x)=0$. These coefficients are bounded as $x\rightarrow 0^\pm$ and $\frac{A_0}{\sqrt{a}}$ and $B_0$ are also bounded as $t\rightarrow 0^+$, as requested in \eqref{Eq_Coefficient_extension}. Moreover, by \cite[Eq.(9.2.14)]{AS72} the derivative of the Hankel function admits the asymptotics
\begin{equation*}
e^{iw}H_\nu^{(2)'}(w)=\mathcal{O}\Big(\frac{1}{\sqrt{w}}\Big),\qquad\text{as }w\rightarrow\infty,\,\Re(w)>0,
\end{equation*}
and it follows again from \cite[Eq.(9.1.4),(9.1.7),(9.1.9),(9.1.27)]{AS72}, that
\begin{equation*}
H_\nu^{(2)'}(w)=\mathcal{O}\Big(\frac{1}{w^{\frac{3}{2}}}\Big),\qquad\text{as }w\rightarrow 0,\,\Re(w)>0.
\end{equation*}
These two asymptotics now ensure the existence of a constnat $D_\nu\geq 0$ with
\begin{equation}\label{Eq_dHankel_estimate}
|H_\nu^{(2)'}(w)|\leq D_\nu\Big(\frac{1}{\sqrt{|w|}}+\frac{1}{|w|^{\frac{3}{2}}}\Big)e^{\Im(w)},\qquad\Re(w)>0.
\end{equation}
Due to the explicit form \eqref{Eq_Gtilde_centrifugal_derivatives_Hankel} of the derivatives of $\widetilde{G}$ and the estimates \eqref{Eq_Gtilde_estimate_Hankel} and \eqref{Eq_dHankel_estimate}, one immediately sees, that also the exponential bounds \eqref{Eq_Gtildex_estimate} and \eqref{Eq_Gtilde_derivative_estimate} are satisfied.
\end{proof}

\begin{proof}[Proof of Proposition \ref{prop_Example_centrifugal_Bessel}]
First of all, for every $t>0$, $x\in\dot{\mathbb{R}}$ the function $G(t,x,\,\cdot\,)$ holomorphically extends to $\Set{z\in\mathbb{C} | \Re(z)\neq 0}$ by
\begin{equation*}
G(t,x,z)=\frac{\Theta(\pm x)\sqrt{xz}}{2i^{\nu+1}t}e^{-\frac{x^2+z^2}{4it}}J_\nu\Big(\frac{xz}{2t}\Big),\qquad\pm\Re(z)>0.
\end{equation*}
In particular, $G(t,x,\,\cdot\,)$ is holomorphic on the double sector $S_\alpha$, for any $\alpha\in(0,\frac{\pi}{2})$. Moreover, the decomposition \eqref{Eq_G_decomposition} is satisfied using $a(t)=\frac{1}{4t}$ and
\begin{equation*}
\widetilde{G}(t,x,z)=\frac{\Theta(\pm x)\sqrt{xz}}{2i^{\nu+1}t}e^{-\frac{xz}{2it}}J_\nu\Big(\frac{xz}{2t}\Big),\qquad\pm\Re(z)>0.
\end{equation*}
Now we verify the properties (i) -- (iii) of Assumption \ref{ass_Greensfunction}.

\begin{enumerate}
\item[(i)] It is obvious, that for fixed $z\in\mathbb{C}$ with $\Re(z)\neq 0$, we have $G(\,\cdot\,,\,\cdot\,,z)\in\AC_{1,2}((0,T)\times\dot{\mathbb{R}})$. In order to see, that it is a solution of \eqref{Eq_G_Schroedinger}, we explicitly calculate the derivatives
\begin{align}
\frac{\partial}{\partial x}\widetilde{G}(t,x,z)&=\Big(\frac{1}{2x}-\frac{z}{2it}\Big)\widetilde{G}(t,x,z)+\frac{\Theta(\pm x)z\sqrt{xz}}{4i^{\nu+1}t^2}e^{-\frac{xz}{2it}}J_\nu'\Big(\frac{xz}{2t}\Big), \notag \\
\frac{\partial^2}{\partial x^2}\widetilde{G}(t,x,z)&=\Big(\frac{\nu^2-\frac{1}{4}}{x^2}-\frac{z}{2itx}-\frac{z^2}{2t^2}\Big)\widetilde{G}(t,x,z)+\frac{\Theta(\pm x)z^2\sqrt{xz}}{4i^\nu t^3}e^{-\frac{xz}{2it}}J_\nu'\Big(\frac{xz}{2t}\Big), \label{Eq_Gtilde_centrifugal_derivatives_Bessel} \\
\frac{\partial}{\partial t}\widetilde{G}(t,x,z)&=\Big(\frac{xz}{2it^2}-\frac{1}{t}\Big)\widetilde{G}(t,x,z)-\frac{\Theta(\pm x)(xz)^{\frac{3}{2}}}{4i^{\nu+1}t^3}e^{-\frac{xz}{2it}}J_\nu'\Big(\frac{xz}{2t}\Big), \notag
\end{align}
where for the second spatial derivative we used the Bessel differential equation
\begin{equation*}
w^2J_\nu''(w)+wJ_\nu'(w)+(w^2-\nu^2)J_\nu(w)=0,\qquad\Re(w)>0.
\end{equation*}
Using $\lambda=\nu^2-\frac{1}{4}$, these derivatives now satisfy 
\begin{equation*}
i\frac{\partial}{\partial t}\widetilde{G}(t,x,z)=\Big(-\frac{\partial^2}{\partial x^2}+\frac{x-z}{it}\frac{\partial}{\partial x}+\frac{1}{2it}+\frac{\lambda}{x^2}\Big)\widetilde{G}(t,x,z),
\end{equation*}
which is equivalent to \eqref{Eq_G_Schroedinger} using the decomposition \ref{Eq_G_decomposition}. From \cite[Eq.(9.1.7)]{AS72} and the fact that $\nu>\frac{1}{2}$ we get the limit
\begin{equation*}
J_\nu(w)\rightarrow 0,\qquad\text{as }w\rightarrow 0,\,\Re(w)>0.
\end{equation*}
Hence we conclude for every $t>0$, $\pm y>0$ the boundary value \eqref{Eq_G_transmission} by
\begin{equation*}
G(t,x,y)=\frac{\Theta(\pm x)\sqrt{xy}}{2i^{\nu+1}t}e^{-\frac{x^2+y^2}{4it}}J_\nu\Big(\frac{xy}{2t}\Big)\overset{x\rightarrow 0}{\longrightarrow}0.
\end{equation*}

\item[(ii)] In order to check the initial condition \eqref{Eq_G_initial}, we fix $x\in\dot{\mathbb{R}}$ and without loss of generality we only consider $x>0$. Let now $x_0>x$ be arbitrary, $\varphi\in\mathcal{C}^\infty([-x_0,x_0])$ and consider the function
\begin{equation*}
\Psi_0(t,x)\coloneqq\int_{-x_0}^{x_0}G(t,x,y)\varphi(y)dy=\frac{1}{2i^{\nu+1}t}\int_0^{x_0}\sqrt{xy}\,e^{-\frac{x^2+y^2}{4it}}J_\nu\Big(\frac{xy}{2t}\Big)\varphi(y)dy,
\end{equation*}
as well as the approximated function
\begin{equation*}
\widetilde{\Psi}_0(t,x)\coloneqq\frac{1}{i^{\nu+1}\sqrt{\pi t}}\int_0^{x_0}e^{-\frac{x^2+y^2}{4it}}\cos\Big(\frac{xy}{2t}-\frac{(2\nu+1)\pi}{4}\Big)\varphi(y)dy.
\end{equation*}
From \cite[Eq.(9.2.1)]{AS72} we get the asymptotics
\begin{equation}\label{Eq_Bessel_asymptotic_inf}
J_\nu(w)-\frac{\sqrt{2}}{\sqrt{\pi w}}\cos\Big(w-\frac{(2\nu+1)\pi}{4}\Big)=e^{|\Im(w)|}\mathcal{O}\Big(\frac{1}{w^{\frac{3}{2}}}\Big),\qquad\text{as }w\rightarrow\infty,\,\Re(w)>0,
\end{equation}
Since moreover, $J_\nu(w)\rightarrow 0$ as $w\rightarrow 0$ by \cite[Eq.(9.1.7)]{AS72}, there exists some $C_\nu\geq 0$ with
\begin{equation*}
\Big|J_\nu(w)-\frac{\sqrt{2}}{\sqrt{\pi w}}\cos\Big(w-\frac{(2\nu+1)\pi}{4}\Big)\Big|\leq\frac{C_\nu}{|w|^{\frac{5}{4}}}e^{|\Im(w)|},\qquad\Re(w)>0,
\end{equation*}
where the exponent $\frac{5}{4}$ is chosen such that in the following estimate the term $t^{\frac{1}{4}}(xy)^{-\frac{3}{4}}$ appears, which on the one hand vanishes in the limit $t\rightarrow 0^+$, but is still integrable at $y=0^+$. With this inequality we can now estimate the error of the approximate function by
\begin{align*}
|\Psi_0(t,x)-\widetilde{\Psi}_0(t,x)|&\leq\frac{1}{2t}\int_0^{x_0}\sqrt{xy}\,\Big|J_\nu\Big(\frac{xy}{2t}\Big)-\frac{2\sqrt{t}}{\sqrt{\pi xy}}\cos\Big(\frac{xy}{2t}-\frac{(2\nu+1)\pi}{4}\Big)\Big||\varphi(y)|dy \\
&\leq C_\nu(2t)^{\frac{1}{4}}\int_0^{x_0}\frac{1}{(xy)^{\frac{3}{4}}}|\varphi(y)|dy \\
&\leq\frac{4C_\nu(2tx_0)^{\frac{1}{4}}}{x^{\frac{3}{4}}}\Vert\varphi\Vert_\infty.
\end{align*}
Since the right hand side converges to zero as $t\rightarrow 0^+$, we get
\begin{equation}\label{Eq_Psi_tildePsi_limit}
\lim\limits_{t\rightarrow 0^+}\Psi_0(t,x)=\lim\limits_{t\rightarrow 0^+}\widetilde{\Psi}_0(t,x),
\end{equation}
and we reduced the problem \eqref{Eq_G_initial} to the initial value of $\widetilde{\Psi}_0(t,x)$. Writing the cosine as an exponential function we can split up the integral as
\begin{align*}
\widetilde{\Psi}_0(t,x)&=\frac{1}{2i^{\nu+1}\sqrt{\pi t}}\int_0^{x_0}e^{-\frac{x^2+y^2}{4it}}\Big(e^{i\frac{xy}{2t}-i\frac{(2\nu+1)\pi}{4}}+e^{-i\frac{xy}{2t}+i\frac{(2\nu+1)\pi}{4}}\Big)\varphi(y)dy \\
&=\frac{1}{2\sqrt{i\pi t}}\int_0^{x_0}\Big((-1)^{\nu+\frac{1}{2}}e^{-\frac{(x+y)^2}{4it}}+e^{-\frac{(x-y)^2}{4it}}\Big)\varphi(y)dy.
\end{align*}
Using the derivative $\frac{d}{dw}\erf(w)=\frac{2}{\sqrt{\pi}}e^{-w^2}$ of the error function and applying integration by parts, one can rewrite this integral as
\begin{align*}
\widetilde{\Psi}_0(t,x)&=\frac{1}{2}\int_0^{x_0}\frac{d}{dy}\Big((-1)^{\nu+\frac{1}{2}}\erf\Big(\frac{x+y}{2\sqrt{it}}\Big)-\erf\Big(\frac{x-y}{2\sqrt{it}}\Big)\Big)\varphi(y)dy \\
&=\frac{(-1)^{\nu+1}}{2}\bigg(\erf\Big(\frac{x+x_0}{2\sqrt{it}}\Big)\varphi(x_0)-\erf\Big(\frac{x}{2\sqrt{it}}\Big)\varphi(0)-\int_0^{x_0}\erf\Big(\frac{x+y}{2\sqrt{it}}\Big)\varphi'(y)dy\bigg) \\
&\quad-\frac{1}{2}\bigg(\erf\Big(\frac{x-x_0}{2\sqrt{it}}\Big)\varphi(x_0)-\erf\Big(\frac{x}{2\sqrt{it}}\Big)\varphi(0)-\int_0^{x_0}\erf\Big(\frac{x-y}{2\sqrt{it}}\Big)\varphi'(y)dy\bigg).
\end{align*}
Apply now the limit $t\rightarrow 0^+$ and carrying it inside the integral is allowed since the integrand is uniformly bounded. Using also $0<x<x_0$ as well as the limit $\lim_{w\rightarrow\pm\infty}\erf(\frac{w}{\sqrt{i}})=\pm 1$ of the error function gives the initial value
\begin{align*}
\lim\limits_{t\rightarrow 0^+}\widetilde{\Psi}_0(t,x)&=\frac{(-1)^{\nu+1}}{2}\bigg(\varphi(x_0)-\varphi(0)-\int_0^{x_0}\varphi'(y)dy\bigg) \\
&\quad-\frac{1}{2}\bigg(-\varphi(x_0)-\varphi(0)-\int_0^{x_0}\sgn(x-y)\varphi'(y)dy\bigg)=\varphi(x).
\end{align*}
Together with \eqref{Eq_Psi_tildePsi_limit}, this proves the initial value \eqref{Eq_G_initial}.

\item[(iii)] By the asymptotics \eqref{Eq_Bessel_asymptotic_inf} and since $J_\nu(w)\rightarrow 0$ as $w\rightarrow 0$ by \cite[Eq.(9.1.7)]{AS72}, there exists $D_\nu\geq 0$ with
\begin{equation*}
|J_\nu(w)|\leq\frac{D_\nu}{\sqrt{|w|}}e^{|\Im(w)|},\qquad\Re(w)>0.
\end{equation*}
With this estimate, the absolute value of $\widetilde{G}$ can be estimated by
\begin{equation*}
|\widetilde{G}(t,x,z)|\leq\frac{\Theta(\pm x)\sqrt{|xz|}}{2t}e^{-\frac{x\Im(z)}{2t}}\Big|J_\nu\Big(\frac{xz}{2t}\Big)\Big|\leq\frac{D_\nu}{\sqrt{2t}},\qquad\pm z\in S_\alpha^+,
\end{equation*}
where we used that $|x\Im(z)|=x\Im(z)$ since $\pm z\in S_\alpha^+$ and $\pm x>0$. Hence $\widetilde{G}$ satisfies the bound \eqref{Eq_Gtilde_estimate} with the coefficients $A_0(t,x)=\frac{D_\nu}{\sqrt{2t}}$ and $B_0(t,x)=0$. These coefficients are bound as $x\rightarrow 0^\pm$ and also $\frac{A_0}{\sqrt{a}}$ and $B_0$ are bounded as $t\rightarrow 0^+$, as requested in \eqref{Eq_Coefficient_extension}. Moreover, it follows from \cite[Eq.(9.1.10),(9.1.27),(9.2.11)]{AS72} as well as $\nu>\frac{1}{2}$, that there exists some $E_\nu\geq 0$ such that the derivative of the Bessel function is bounded by
\begin{equation}\label{Eq_dBessel_estimate}
|J_\nu'(w)|\leq\frac{E_\nu}{\sqrt{|w|}}e^{|\Im(w)|},\qquad\Re(w)>0.
\end{equation}
Due to the explicit form \eqref{Eq_Gtilde_centrifugal_derivatives_Bessel} of the derivatives of $\widetilde{G}$ and the additional estimate \eqref{Eq_dBessel_estimate} of the derivative of the Bessel function, one immediately sees, that also the exponential bounds \eqref{Eq_Gtildex_estimate} and \eqref{Eq_Gtilde_derivative_estimate} are satisfied. \qedhere
\end{enumerate}
\end{proof}

\subsection{Arbitrary point interactions}\label{sec_Arbitrary_point_interactions}

In this subsection we consider the vanishing classical potential $V(t,x)=0$, $x\in\dot{\mathbb{R}}$, and allow all possible self-adjoint singular interactions at the origin $x=0^\pm$. In particular, the Dirac $\delta$- and $\delta'$-potential and boundary conditions of Dirichlet-, Neumann- and Robin-type are included, see \cite[Section 3]{ABCS21_1}. Although the time persistence problem with respect to those point interactions was already considered in \cite{ABCS21_1}, this section shows that these distributional potentials are also covered by the general theory of Section \ref{sec_Schroedinger_equation_on_R0}.

\medskip

There are various ways to describe the complete family of self-adjoint interface conditions, but for our purposes it is convenient to use the parametrization
\begin{equation}\label{Eq_Interface_condition}
(I-J)\vvect{\Psi(t,0^+)}{\Psi(t,0^-)}=i(I+J)\vvect{\Psi_x(t,0^+)}{-\Psi_x(t,0^-)},
\end{equation}
where $I$ is the $2\times 2$ identity matrix and $J$ is some arbitrary $2\times 2$ unitary matrix, which in turn can be represented by
\begin{equation}\label{Eq_J}
J=e^{i\phi}\mmatrix{\alpha}{-\overline{\beta}}{\beta}{\overline{\alpha}},
\end{equation}
with parameters $\phi\in[0,\pi)$ and $\alpha,\beta\in\mathbb{C}$ satisfying $|\alpha|^2+|\beta|^2=1$. In order to write down the corresponding Green's function, we start by defining the entire function
\begin{equation}\label{Eq_Lambda}
\Lambda(z)\coloneqq e^{z^2}(1-\erf(z)),\qquad z\in\mathbb{C},
\end{equation}
which is a modification of the well known error function $\erf(z)=\frac{2}{\sqrt{\pi}}\int_0^ze^{-w^2}dw$. The Green's function is now given by
\begin{align}
G(t,x,y)\coloneqq&\mu_+^{(x,y)}\Lambda\Big(\frac{|x|+|y|}{2\sqrt{it}}+\omega_+\sqrt{it}\Big)e^{-\frac{(|x|+|y|)^2}{4it}}+\mu_-^{(x,y)}\Lambda\Big(\frac{|x|+|y|}{2\sqrt{it}}+\omega_-\sqrt{it}\Big)e^{-\frac{(|x|+|y|)^2}{4it}} \notag \\
&+\frac{\mu_0^{(x,y)}}{2\sqrt{i\pi t}}e^{-\frac{(|x|+|y|)^2}{4it}}+\frac{1}{2\sqrt{i\pi t}}e^{-\frac{(x-y)^2}{4it}},\qquad t>0,\,x,y\in\dot{\mathbb{R}}. \label{Eq_G_point}
\end{align}
The values of the coefficients $\mu_\pm^{(x,y)}$, $\mu_0^{(x,y)}$ and $\omega_\pm$ will be specified in terms of the unitary matrix $J$ in the following. To do so, it is convenient to use
\begin{subequations}\label{Eq_eta}
\begin{align}
\eta^{(x,y)}&\coloneqq\frac{1}{\sqrt{1-\Re(\alpha)^2}}\left\{\begin{array}{ll} -\Im(\alpha), & \text{if }x,y>0, \\ -i\overline{\beta}, & \text{if }x>0,\,y<0, \\ i\beta, & \text{if }x<0,\,y>0, \\ \Im(\alpha), & \text{if }x,y<0, \end{array}\right.\qquad\text{if }|\Re(\alpha)|\neq 1, \label{Eq_eta1} \\
\eta^{(x,y)}&\coloneqq 0, \hspace{7.32cm}\text{if }|\Re(\alpha)|=1, \label{Eq_eta2}
\end{align}
\end{subequations}
and distinguish the following three cases.

\medskip

\textbf{Case I}: If $\Re(\alpha)\neq-\cos(\phi)$, then
\begin{equation*}
\omega_\pm=\frac{-\sin(\phi)\pm\sqrt{1-\Re(\alpha)^2}}{\cos(\phi)+\Re(\alpha)},\qquad\mu_\pm^{(x,y)}=\frac{\omega_\pm}{2}\big(\Theta(xy)+\eta^{(x,y)}\big),\qquad\mu_0^{(x,y)}=\sgn(xy).
\end{equation*}

\textbf{Case II}: If $\Re(\alpha)=-\cos(\phi)\neq -1$, then $\omega_-=\mu_-^{(x,y)}=0$ and
\begin{equation*}
\omega_+=\cot(\phi),\qquad\mu_+^{(x,y)}=-\frac{\omega_+}{2}\big(\Theta(xy)+\eta^{(x,y)}\big),\qquad\mu_0^{(x,y)}=\eta^{(x,y)}-\Theta(-xy).
\end{equation*}

\textbf{Case III}: If $\Re(\alpha)=-\cos(\phi)=-1$, then $\omega_\pm=\mu_\pm^{(x,y)}=0$ and $\mu_0^{(x,y)}=-1$.

\medskip

\begin{prop}
The Green's function \eqref{Eq_G_point} satisfies the Assumptions \ref{ass_Greensfunction} with respect to the potential $V(t,x)=0$ and the transition matrices $M=I-J$ and $N=i(I+J)$.
\end{prop}

\begin{proof}
First of all, it is already shown \cite[Theorem 2.4]{ABCS21_1} that $G$ satisfies the transmission condition \eqref{Eq_G_transmission}.
Moreover, the Green's function \eqref{Eq_G_point} holomorphically extends to $\Set{z\in\mathbb{C} | \Re(z)\neq 0}$ by
\begin{align}
G(t,x,z)\coloneqq&\mu_+^{(x,\pm)}G_1(t,x,\pm z)+\mu_-^{(x,\pm)}G_1(t,x,\pm z) \notag \\
&+\mu_0^{(x,\pm)}G_0(t,x,\pm z)+G_\text{free}(t,x,z),\qquad\pm\Re(z)>0, \label{Eq_G_decomposition_point}
\end{align}
using the functions
\begin{subequations}\label{Eq_G01free}
\begin{align}
G_0(t,x,z)&=\frac{1}{2\sqrt{i\pi t}}e^{-\frac{(|x|+z)^2}{4it}}, && \Re(z)>0, \label{Eq_G0} \\
G_1(t,x,z;\omega)&\coloneqq\Lambda\Big(\frac{|x|+z}{2\sqrt{it}}+\omega\sqrt{it}\Big)e^{\frac{(|x|+z)^2}{4it}}, && \Re(z)>0, \label{Eq_G1} \\
G_\text{free}(t,x,z)&\coloneqq\frac{1}{2\sqrt{i\pi t}}e^{-\frac{(x-z)^2}{4it}}, && z\in\mathbb{C}. \label{Eq_Gfree}
\end{align}
\end{subequations}
It is already proven in \cite[Lemma 2.2]{ABCS21_1} that all of the three functions \eqref{Eq_G01free} satisfy the Schrödinger equation \eqref{Eq_G_Schroedinger}, and consequently so does $G$, since the coefficients $\mu_+^{(x,\pm)}$, $\mu_-^{(x,\pm)}$, $\mu_0^{(x,\pm)}$ only depend on $\sgn(x)$. Also for the initial condition \eqref{Eq_G_initial} it is sufficient to consider the functions \eqref{Eq_G01free} separately. Fix $x\in\dot{\mathbb{R}}$ choose $x_0>|x|$ arbitrary and $\varphi\in\mathcal{C}^\infty([-x_0,x_0])$. For $G_\text{free}$, we can use the derivative $\frac{d}{dw}\erf(w)=\frac{2}{\sqrt{\pi}}e^{-w^2}$ and apply integration by parts to write the integral as
\begin{align*}
\int_{-x_0}^{x_0}G_\text{free}(t,x,y)\varphi(y)dy&=-\frac{1}{2}\int_{-x_0}^{x_0}\frac{d}{dy}\erf\Big(\frac{x-y}{2\sqrt{it}}\Big)\varphi(y)dy \\
&\hspace{-2cm}=\frac{1}{2}\erf\Big(\frac{x+x_0}{2\sqrt{it}}\Big)\varphi(-x_0)-\frac{1}{2}\erf\Big(\frac{x-x_0}{2\sqrt{it}}\Big)\varphi(x_0)+\frac{1}{2}\int_{-x_0}^{x_0}\erf\Big(\frac{x-y}{2\sqrt{it}}\Big)\varphi'(y)dy.
\end{align*}
Applying the limit $t\rightarrow 0^+$ and using $\lim_{w\rightarrow\pm\infty}\erf(\frac{w}{\sqrt{i}})=\pm 1$, we get the initial value
\begin{equation}\label{Eq_Gfree_initial}
\lim\limits_{t\rightarrow 0^+}\int_{-x_0}^{x_0}G_\text{free}(t,x,y)\varphi(y)dy=\frac{\varphi(-x_0)+\varphi(x_0)}{2}+\frac{1}{2}\int_{-x_0}^{x_0}\sgn(x-y)\varphi'(y)dy=\varphi(x).
\end{equation}
For the function $G_1$ we get
\begin{align*}
\int_{-x_0}^{x_0}G_1(t,x,y;\omega)\varphi(y)dy&=\frac{1}{2}\int_{-x_0}^{x_0}\sgn(y)\Lambda\Big(\frac{|x|+|y|}{2\sqrt{it}}+\omega\sqrt{it}\Big)\varphi(y)\frac{d}{dy}\erf\Big(\frac{|x|+|y|}{2\sqrt{it}}\Big)dy \\
&=-\frac{1}{2}\erf\Big(\frac{|x|+|y|}{2\sqrt{it}}\Big)\Lambda\Big(\frac{|x|+|y|}{2\sqrt{it}}+\omega\sqrt{it}\Big)\varphi(y)\Big|_{y=-x_0}^0 \\
&\quad+\frac{1}{2}\erf\Big(\frac{|x|+|y|}{2\sqrt{it}}\Big)\Lambda\Big(\frac{|x|+|y|}{2\sqrt{it}}+\omega\sqrt{it}\Big)\varphi(y)\Big|_{y=0}^{x_0} \\
&\quad-\frac{1}{2}\int_{-x_0}^{x_0}\sgn(y)\erf\Big(\frac{|x|+|y|}{2\sqrt{it}}\Big)\frac{d}{dy}\Big(\Lambda\Big(\frac{|x|+|y|}{2\sqrt{it}}+\omega\sqrt{it}\Big)\varphi(y)\Big)dy.
\end{align*}
From \cite[Lemma 3.1]{ABCS19} we know that $\Lambda(z)=\mathcal{O}\big(\frac{1}{z}\big)$ and also $\Lambda'(z)=\mathcal{O}\big(\frac{1}{z}\big)$, as $z\rightarrow\infty$, $\Re(z)>0$. Hence all terms in the above representation vanish in the limit $t\rightarrow 0^+$, which gives the initial value
\begin{equation}\label{Eq_G1_initial}
\lim\limits_{t\rightarrow 0^+}\int_{-x_0}^{x_0}G_1(t,x,y;\omega)\varphi(y)dy=0.
\end{equation}
The same calculations with the absence of the function $\Lambda$ also leads to the initial value
\begin{equation}\label{Eq_G0_initial}
\lim\limits_{t\rightarrow 0^+}\int_{-x_0}^{x_0}G_0(t,x,y)\varphi(y)dy=0.
\end{equation}
Using now the three limits \eqref{Eq_Gfree_initial}, \eqref{Eq_G1_initial} and \eqref{Eq_G0_initial} in the decomposition \eqref{Eq_G_decomposition_point} gives the initial value \eqref{Eq_G_initial}.

\medskip

In order to derive the estimate \eqref{Eq_Gtilde_estimate}, it is sufficient to do these estimates for the three functions \eqref{Eq_G01free}, or rather for their reduced representations
\begin{equation*}
\widetilde{G}_j(t,x,z)\coloneqq e^{\frac{(x-z)^2}{4it}}G_j(t,x,z),\qquad j\in\{0,1,\text{free}\},
\end{equation*}
separately. However, this is already done in \cite[Lemma 2.3]{ABCS21_1}, and we conlude the estimates
\begin{align*}
|\widetilde{G}_0(t,x,z)|&\leq\frac{1}{2\sqrt{\pi t}}, && z\in S_\alpha^+, \\
|\widetilde{G}_1(t,x,z;\omega)|&\leq\Lambda\Big(\frac{\omega\sqrt{t}}{\sqrt{2}}\Big), && z\in S_\alpha^+, \\
|\widetilde{G}_\text{free}(t,x,z)|&=\frac{1}{2\sqrt{i\pi t}}, && z\in S_\alpha,
\end{align*}
Hence the bound \eqref{Eq_Gtilde_estimate} es obviously satisfied with coefficients $A_0$ and $B_0$, which are bounded as $x\rightarrow 0^\pm$ and also satisfy the bounds \eqref{Eq_Coefficient_extension} in the limit $t\rightarrow 0^+$. Moreover, the estimates of the first spatial derivatives can also be concluded from \cite[Lemma 2.2 \& Lemma 2.3]{ABCS21_1} and are given by
\begin{align*}
\Big|\frac{\partial}{\partial x}\widetilde{G}_0(t,x,z)\Big|&\leq\frac{|z|}{2\sqrt{\pi}\,t^{\frac{3}{2}}}, \\
\Big|\frac{\partial}{\partial x}\widetilde{G}_1(t,x,z;\omega)\Big|&\leq\Big(\frac{|x|+|z|}{2t}+|\omega|\Big)\Lambda\Big(\frac{\omega\sqrt{t}}{\sqrt{2}}\Big)+\frac{1}{\sqrt{\pi t}}, \\
\Big|\frac{\partial}{\partial x}\widetilde{G}_\text{free}(t,x,z)\Big|&=0.
\end{align*}
If we additionally use that $|z|\leq e^{|z|-1}$, the estimate \eqref{Eq_Gtildex_estimate} is satisfied with coefficients $A_1$ and $B_1$, which are bounded as $x\rightarrow 0^\pm$. Finally, also the exponential bounds of the second spatial and the time derivatives \eqref{Eq_Gtilde_derivative_estimate} of $\widetilde{G}_i$, $i\in\{0,1,\text{free}\}$ follow from \cite[Lemma 2.2 \& Lemma 2.3]{ABCS21_1}.
\end{proof}

\end{document}